\documentclass[onecolumn]{IEEEtran}
\pdfoutput=1
\usepackage{amsmath,amssymb,amsthm,mathrsfs}
\usepackage{graphicx,subfigure}
\usepackage{tikz}
    \usetikzlibrary{patterns}
    \usetikzlibrary{decorations}
    \usetikzlibrary{decorations.pathreplacing}
    \usetikzlibrary{arrows}
\usepackage{pgfplots}
\usepackage[numbers,sort&compress,square]{natbib}
\allowdisplaybreaks
\newtheorem{theorem}{Theorem}
\newtheorem{lemma}[theorem]{Lemma}
\newtheorem{claim}[theorem]{Claim}

\newtheorem{proposition}[theorem]{Proposition}
\theoremstyle{definition}
     
     \newtheorem{definition}{Definition}
     
     \newtheorem{remark}{Remark}
\theoremstyle{remark}

\newcommand{\plus}[1]{{[#1]}^{\scalebox{0.6}{\!+}}}
\newcommand{\sender}{\mathsf{S}}
\newcommand{\user}{\mathsf{U}}
\newcommand{\msg}{\vec{w}}
\newcommand{\sinfo}{\vec{\phi}}

\newcommand{\bern}{\mathrm{Bernoulli}\left(\frac{1}{2}\right)}

\DeclareMathOperator{\E}{E}

\DeclareMathOperator*{\maximize}{maximize\ }

\DeclareMathOperator*{\subjectto}{subject\ to\ }

\begin{document}

\title{Blind Index Coding}
\author{
   \IEEEauthorblockN{
     David T.H. Kao\IEEEauthorrefmark{1},
     Mohammad Ali Maddah-Ali\IEEEauthorrefmark{2}, and
     A. Salman Avestimehr\IEEEauthorrefmark{3}}\\
   \IEEEauthorblockA{
     \IEEEauthorrefmark{1}University of Southern California, Los Angeles, CA, USA \qquad\quad \IEEEauthorrefmark{2}Bell Labs, Holmdel, NJ, USA}
\thanks{The research of A.S. Avestimehr and D.T.H. Kao was supported by NSF Grants CAREER 1408639, CCF-1408755, NETS-1419632, EARS-1411244, ONR award N000141310094, and research grants from Intel and Verizon via the 5G project, and was completed while D.T.H. Kao was part of the Universtiy of Southern California. D.T.H. Kao is now a part of Google Inc.
Parts of this work were presented in \cite{KMA:isit2015} and \cite{KMA:icc2015}.}
}

\maketitle

\begin{abstract}
We introduce the \emph{blind index coding} (BIC) problem, in which a single sender communicates distinct messages to multiple users over a shared channel. Each user has partial knowledge of each message as side information. However, unlike classic index coding, in BIC, the sender is uncertain of what side information is available to each user. In particular, the sender only knows the amount of bits in each user's side information but not its content. This problem can arise naturally in caching and wireless networks. In order to blindly exploit side information in the BIC problem, we develop a hybrid coding scheme that XORs uncoded bits of a subset of messages with random combinations of bits from other messages. This scheme allows us to strike the right balance between maximizing the transmission rate to each user and minimizing the interference leakage to others. We also develop a general outer bound, which relies on a strong data processing inequality to effectively capture the sender’s uncertainty about the users' side information. Additionally, we consider the case where communication takes place over a shared \emph{wireless} medium, modeled by an erasure broadcast channel, and show that surprisingly, combining repetition coding with hybrid coding improves the achievable rate region and outperforms alternative strategies of coping with channel erasure and while blindly exploiting side information. 
\end{abstract}

\section{Introduction}

In many communication scenarios, users have access to some side information about the messages that are requested by other users. For example, this scenario can arise in caching networks in which caches opportunistically store content that may be requested in the future. It can also arise in wireless networks in which nodes can overhear the signals intended for other nodes over the shared wireless medium~\cite{KMA2014:isit}. However, since there are many possibilities for what each cache can store at a particular time (or for what signals each node can overhear in a wireless network), tracking the exact content of side information at the users can be very challenging. Therefore, it is more suitable to require the server only track the ``amount'' of side information at each user, and not its exact ``content''. Consequently, a natural question is: how can a sender take advantage of knowledge of only the amount of side information to efficiently deliver messages to users?

To understand this problem, and evaluate and isolate the ultimate gain of such side information, we introduce a basic network communication problem with one sender and several users, depicted in Figure~\ref{fig:model}.
The sender communicates a distinct message, $\vec{w}_i$, to each of $K$ users (labeled $i=1,\ldots,K$) over a broadcast communication channel, while each user, $i$, has some prior side information ($\sinfo_{ij}$) about other users' desired messages ($\vec{w}_j$ where $j\neq i$) that it may use to assist in decoding its own desired message. However, the sender does not know the precise side information given to each user (i.e., the sender is \emph{blind}), and it must employ a transmission strategy that only uses knowledge of the probability distributions of $\sinfo_{ij}$, for all $i\neq j$.

We refer to this new formulation as the \emph{blind index coding} (BIC) problem.
Our formulation is a generalization of the classic index coding problem~\cite{BK2006,ByBJK2011,ALSWH2008:focs}, which is a canonical problem in network communication theory and, despite its simple formulation, remains a powerful tool for analyzing many network communication settings (see e.g.,~\cite{NA2015,Jafar2014,MCJ2014,MaN2014,ErSG2010}). The key difference in BIC problems lies in the sender's uncertainty regarding side information: In classic index coding, precise knowledge of side information is used by the sender to create transmission strategies that treat message bits differently depending on whether they are within or not within side information at each particular user~\cite{ABKSW2013:isit,Ong2014:isit}. However, in BIC the sender is unable to distinguish between such message bits, and thus transmission must ``blindly'' exploit knowledge of the only the amount of side information. As we will see, this minor difference significantly changes the technical challenges of the problem.

The main question that we investigate in this paper is \emph{``To what extent and using what techniques can we blindly exploit such side information?''} To that end, after formally introducing the BIC problem, our first contribution is the development of 
a class of \emph{hybrid coding} schemes, which XOR random linear combinations of bits from one subset of messages with uncoded bits from a disjoint subset of messages. In these hybrid codes, the sender XORs uncoded bits in order to probabilistically exploit side information already available at users. We first provide an example to show that this approach can outperform random coding and in fact sometimes achieve capacity. We then construct an general achievable scheme for three users based on this approach and determine the achievable symmetric rate. 

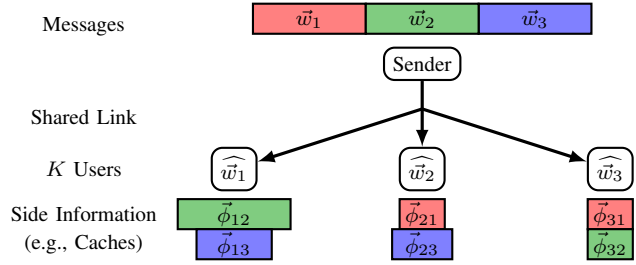
\begin{figure}\centering
\begin{tikzpicture}[font=\footnotesize]
\node (M) at (-4.5,9.5) []{Messages};
\node (C) at (-4.5,8.3) []{Shared Link};
\node (U) at (-4.5,7.6) []{$K$ Users};
\node (SI) at (-4.5,7) []{Side Information};
\node (SI) at (-4.5,6.6) []{(e.g., Caches)};

\node (S) at (0,9) [thick, rounded corners,draw]{Sender};

\filldraw[thick,fill=blue!50!white] (0.75,9.4) rectangle (2.25,9.8);
\filldraw[thick,fill=red!50!white] (-2.25,9.4) rectangle (-0.75,9.8);
\filldraw[thick,fill=green!60!black!50!white] (-0.75,9.4) rectangle (0.75,9.8);

\draw (-1.5,9.6) node {$\vec{w}_1$};
\draw (1.5,9.6) node {$\vec{w}_3$};
\draw (0,9.6) node {$\vec{w}_2$};

\node (U1) at (-2.5,7.6) [thick, rounded corners,draw]{$\widehat{\vec{w}_1}$};
\node (U2) at (0,7.6) [thick, rounded corners,draw]{$\widehat{\vec{w}_2}$};
\node (U3) at (2.5,7.6) [thick, rounded corners,draw]{$\widehat{\vec{w}_3}$};

\draw[thick,fill=green!60!black!50!white] (-3.25,6.8) rectangle (-1.75,7.2);
\draw[thick,fill=blue!50!white] (-3,6.4) rectangle (-2,6.8);
\draw (-2.5,7) node {$\sinfo_{12}$};
\draw (-2.5,6.6) node {$\sinfo_{13}$};

\draw[thick,fill=red!50!white] (-0.3,6.8) rectangle (0.3,7.2);
\draw[thick,fill=blue!50!white] (-0.4,6.4) rectangle (0.4,6.8);
\draw (0,7) node {$\sinfo_{21}$};
\draw (0,6.6) node {$\sinfo_{23}$};

\draw[thick,fill=red!50!white] (2.2,6.8) rectangle (2.8,7.2);
\draw[thick,fill=green!60!black!50!white] (2.2,6.4) rectangle (2.8,6.8);
\draw (2.5,7) node {$\sinfo_{31}$};
\draw (2.5,6.6) node {$\sinfo_{32}$};

\draw[very thick,-latex] (S) -- (U2);
\draw[very thick,-latex] (0,8.4) -- (U1);
\draw[very thick,-latex] (0,8.4) -- (U3);
\end{tikzpicture}\vspace{-2ex}
\caption{A $K$-user blind index coding problem (e.g., $K=3$) depicted as a sender-user network with user caches. User~$i$, for $i=1,2,3$, desires message $\vec{w}_i$ and may use side information about other messages to facilitate decoding; $\sinfo_{ij}$ denotes the side information that User~$i$ has about Message $\vec{w}_j$. The sender only has knowledge of the distribution of $\sinfo_{ij}$, and not its precise realization. Notice the \emph{amount} of side information available may vary across users and messages.\vspace{-5ex}}\label{fig:model}\end{figure}

In order to evaluate the efficacy of our scheme, as well as to gain further intuition beyond three users, our second contribution is the development of a new outer bound on the capacity region. An essential aspect of our outer bound is its utilization of a strong data processing inequality~\cite{AGKN2014:isit} which captures the inability of the sender to distinguish between bits of a message known or unknown to a given user prior to transmission. We demonstrate that our converse is tight in two special cases: namely, the two-user and symmetric $K$-user BIC (where all users have the same amount of knowledge about undesired messages). In both cases a simple achievable scheme based on random coding suffices to achieve the entire capacity region. As we move beyond these special cases to the general BIC setting, we confirm that, at least for some problem settings, our three-user hybrid coding scheme can meet the symmetric capacity upper bound. Finally, we numerically evaluate our new hybrid coding scheme and outer bounds relative to existing achievable schemes.

In our final contribution, we further consider the BIC problem in a wireless setting, specifically studying how lossy sender-to-user links can affect schemes to blindly exploit side information and the resulting achievable rates. Interestingly, we demonstrate that in addition to hybrid coding (where XORing uncoded bits of a subset of messages with random combinations of the others played a key role), quite surprisingly, XORing the same uncoded bits more than once (i.e., \emph{repetition of uncoded bits}) can increase the achievable rate.  
Equipped with this observation, we then proceed to construct a coding scheme that leverages both hybrid codes and repetition of uncoded message bits in order to establish an achievable rate region, and we demonstrate numerically that such a scheme offers a strict improvement in achievable rate over conventional approaches.

To summarize, the main contributions are as follows:
\begin{enumerate}
    \item We introduce the \emph{Blind Index Coding} problem, which generalizes classic Index Coding by considering uncertainty (blindness) at the sender about side information given to users.
    \item We propose a class of \emph{hybrid coding} schemes, which XOR random linear combinations from one subset of messages with uncoded bits of another subset.
    \item We derive a novel outer bound on the capacity region of BIC problems which leverages a strong data processing inequality to account for the blindness of the sender.
    \item We further generalize the problem to better model wireless settings by studying how lossy sender-to-user links affect the efficacy of the hybrid coding schemes, and we find that repetition coding can enhance the performance of hybrid codes.
\end{enumerate}

This remainder of the paper is organized in the following way.
In Section~\ref{sec:problem} we formally state the BIC problem first for error-free broadcast and then for broadcast over lossy channels. 
In Section~\ref{sec:example} we motivate both the ideas behind hybrid coding and our outer bound using a simple example, for which the inner and outer bounds meet. 
In Section~\ref{sec:achieve}, we define a hybrid coding scheme and study the achievable symmetric rate for the three-user BIC, in Section~\ref{sec:converse}, we state and prove the general outer bound for BIC problems, and in Section~\ref{sec:num} we numerically compare achieved rates to the derived outer bounds.
In Section~\ref{sec:BICW} we consider blind index coding when the sender-to-user links occur over wireless channels. 
Concluding remarks and open questions are presented in Section~\ref{sec:concl}.

\section{The Blind Index Coding Problem} \label{sec:problem}

In this section, we formally define the Blind Index Coding problem by stating the network and side information models, and formalizing the notion of capacity.

\subsubsection*{Network model} 
In a BIC problem, as shown in Figure~\ref{fig:model}, $K$ users each request a message from a sender; i.e., User~$i$, for $i=1,\ldots,K$, desires the $m_i$-bit message $\msg_i$, which is drawn uniformly from a space $\{0,1\}^{m_i}$. Each user, $i$, has access to side information, $\sinfo_{ij}$, (whose form is described later) about each message $\msg_j$ except the one it desires (i.e., for all $j\neq i$). The sender aims to communicate all messages to the respective users via a common error-free channel. The goal of the problem is to design a scheme that maps messages to a channel input vector, $\vec{x}$, of minimum length, such that each user can decode its desired message.


\subsubsection*{Side information model} 
In a blind index coding problem, each side information signal, $\sinfo_{ij}$, is a random fraction of the bits that make up the message, $\msg_j$. We assume that the sender is ``blind'' in the sense that it is only aware of the \emph{average number of bits} in each side information signal.

More specifically, we can model the side information in the following way. Let $\vec{g}_{ij}$ be a length-$m_j$ binary vector drawn i.i.d from a Bernoulli$(1-\mu_{ij})$ distribution. Side information $\sinfo_{ij} = (\phi_{ij}[1],\phi_{ij}[2],\ldots,\phi_{ij}[m_{j}])$ is such that, for $\ell=1,\ldots,m_j$,
\begin{align}
    \phi_{ij}[\ell] = g_{ij}[\ell]{w}_j[\ell].
\end{align}
User~$i$ knows $\vec{g}_{ij}$ for all $j\neq i$, however the sender is only aware of parameters, $\{\mu_{ij}\}$, which govern the probabilistic behavior of the side information. Note that the side information model is equivalent to either 1) randomly sampling bits of a message, or 2) passing a message through a side information channel which is an erasure channel.

\begin{remark}
The key difference between the BIC problem formulation and a classic index coding problem of~\cite{BK2006,ALSWH2008:focs} lies in the uncertainty in message bits given as side information. Notably, if we consider the scenario  where $\mu_{ij}\in\{0,1\}$ for all $i,j$, then side information availability is deterministic and known to the sender, and our formulation is identical to~\cite{ALSWH2008:focs}. Thus, BIC generalizes classic index coding.
\end{remark}
\begin{remark}
Index coding problems with transformed and random side information were considered in~\cite{LDH2015} and~\cite{HL2012:isit} respectively, but in both cases it was assumed that the side information is known to the sender. Another related problem is described in~\cite{BF2013:isit} where pliable users are considered: Users express no specificity in messages demanded and thus which messages to send are uncertain. Interestingly, in the cases of pliable index coding with known solutions, either canonical random coding or uncoded transmission strategies were sufficient.
\end{remark}

\subsubsection*{Capacity Region}
We now consider a BIC problem with $K$ users and side information parameters $\{\mu_{ij}\}$ as defined above. For this problem, a $(r_1,r_2,\ldots,r_K)$ scheme with block length $n$ consists of an encoding function and $K$ decoding functions. 

The encoding function,
\mbox{$f_{\mathsf{enc}}^{(n)}:\prod_{j=1}^K\{0,1\}^{m_j} \rightarrow \{0,1\}^n$},
uses the knowledge of $\{\mu_{ij}\}$ for all $j\neq i$ to map each of $K$ messages (with message $\msg_j$ consisting of $m_j$ bits such that $\lim_{n\rightarrow\infty}\frac{m_{j}}{n} = r_j$) onto a length-$n$ binary vector, $\vec{x}$, that is broadcast to all $K$ users using $n$ channel uses. We reemphasize that the encoding function relies only on the side information parameters, $\{\mu_{ij}\}$, and not the side information signals, $\{\sinfo_{ij}\}$.

The decoding function applied by User~$i$, 
\mbox{$f_{\mathsf{dec},i}^{(n)}: \{0,1\}^{n}\times\prod_{j\neq i}\{0,1\}^{m_j}\times\{0,1\}^{m_j} \rightarrow \{0,1\}^{m_i}$},
maps the broadcast signal, $\vec{x}$, as well as $K-1$ side information vector pairs, $(\sinfo_{ij},\vec{g}_{ij})$ for all $j\neq i$, to an estimate of its desired message, $\widehat{\msg}_i$.

We say that a rate tuple $(r_1,\ldots,r_K)$ is \emph{achievable} if there exists a sequence of $(r_1,\ldots,r_K)$ coding schemes with increasing block length, $n$, such that for every $i\in\{1,\ldots,K\}$
\begin{align}
\lim_{n\rightarrow\infty} \Pr\left[\widehat{\msg}_i\neq \msg_i\right] = 0.
\end{align}
The capacity region is defined as the closure of the set of all rate tuples $(r_1,\ldots,r_K)$ that are achievable. 

The goal of this paper is to study the capacity region of the BIC problem. As we show later in Proposition~\ref{prop:2u}, the capacity region of a 2-user BIC problem is easy to characterize. Thus, in order to gain a better intuition on BIC problems beyond 2 users, we focus in particular on the 3-user BIC problem.

\section{Motivating Example} \label{sec:example}

In this section we motivate both the proposed coding schemes and outer bound using a simple, concrete example.
Consider a BIC with three users (i.e., $K=3$) and where Users~2 and 3 have \emph{full side information} about other users' messages, while User~1 only knows a third of each of $\msg_2$ and $\msg_3$ (i.e., $\mu_{12}=\mu_{13}=\frac{2}{3}$ and $\mu_{21}=\mu_{23}=\mu_{31}=\mu_{32}=0$). We focus on this specific BIC problem because in this scenario the sender is blind \emph{only about side information at User~1} and therefore we can focus on the impact of blindness regarding just one user.

For this particular BIC problem, we will determine the symmetric capacity (i.e., the maximum rate $r$ such that $r_1=r_2=r_3=r$ is achievable) by assuming the lengths of all messages are the same (i.e., $m_1=m_2=m_3=m$ where $m$ is large), proposing a scheme, and introducing a method to bound the capacity region.

For the sake of comparison, we will first establish a baseline achievable symmetric rate by considering random coding, an often-used approach to coding in the presence of uncertain side information. For example, one natural scheme would be to send random linear combinations (RLC) of all message bits (i.e., parity bits to supplement side information) over the shared channel, until each user has a sufficient number of linearly independent equations (including side information) to decode all of the messages. 
For this example, by sending $m(1+\mu_{12}+\mu_{13}) +o(m) = \frac{7m}{3}+o(m)$ random parities, each user has at least $3m+o(m)$ equations for $3m$ unknowns, meaning that with high probability each user can linearly decode all three messages: conventional random coding achieves $r_{sym}=\frac{3}{7}$.\footnote{In the subsequent explanation, we omit the $o(m)$ to simplify the exposition.}

Notice first that we can achieve better by first \emph{grouping} 2 and 3 and sending each bit of $\msg_2$ XORed with a distinct bit of $\msg_3$ (this requires exactly $m$ transmissions). From these transmissions, Users~2 and 3 can use side information to remove the other message and decode their desired message. Then, by sending $\msg_1$ orthogonally in time (requiring another $m$ transmissions) User~1 receives its desired message. In other words by treating subsets of messages differently, we achieved $r_{sym}=\frac{1}{2}$. 

We now demonstrate how to further improve the transmission strategy by constructing a ``hybrid coding scheme'' using a combination of uncoded bits and randomly coded parities to go beyond the rate of $\frac{1}{2}$. In these hybrid schemes, during each phase of transmission a subset of messages are randomly coded, and then these are XORed with uncoded bits from another disjoint subset of the messages. 
For this example, we only require two such phases.
In the first phase, each channel input is generated by XORing a random combination of $\msg_1$ bits, a single uncoded $\msg_2$ bit, and , a single uncoded $\msg_3$ bit. Each uncoded bit from both $\msg_2$ and $\msg_3$ are used only once to generate an input, and thus the first phase consists of exactly $m$ channel inputs generated in this manner. Formally, for each $\ell=1,\ldots,m$, the sender broadcasts ${\vec{c}[\ell]}^\top\msg_1 \oplus w_2[\ell] \oplus w_3[\ell]$, where $\vec{c}[\ell]$ is a length-$m$ i.i.d. random binary vector. 
In the second phase, we send $\frac{8m}{9}$ RLCs of only $\msg_1$ bits. 

Notice that with this scheme, if each user decodes its desired message with error probability vanishing as $m$ grows large, we achieve rate of $r_{sym}=\frac{9}{17}$ which is higher than the $\frac{3}{7}$ achieved through conventional random coding and $\frac{1}{2}$ achieved through grouped random coding. We now explain why with this scheme such a rate is achievable by explaining how each user decodes its desired message:
\begin{description}[]
\item[User~1:]\ \ 
    Notice that during the first phase, for each channel input $\ell\in\{1,\ldots,m\}$, there is a probability of $(1-\mu_{12})(1-\mu_{13}) = \frac{1}{9}$ that User 1 knew both $w_2[\ell]$ and $w_3[\ell]$. In such an event, User~1 can cancel $w_2[\ell]\oplus w_3[\ell]$ and received a ``clean'' RLC of $\msg_1$ bits. Therefore, during the first phase User~1 receives (approximately) $\frac{m}{9}$ such RLCs. In the second phase we supplemented this with an additional $\frac{8m}{9}$ RLCs of only $\msg_1$. When combined, at the end of transmission User~1 will be able to identify in total $m$ linearly independent equations describing the $m$ desired bits of $\msg_1$.
\item[User~2:]\ \ 
 User~2 already knows all of $\msg_1$ and $\msg_3$ and therefore can remove their contributions from each channel input of the first phase. Thus, after canceling the undesired message contributions, User~2 receives exactly the $m$ bits of $\msg_2$.
\item[User~3:]\ \ 
User~3 already knows all of $\msg_1$ and $\msg_2$ and therefore can remove their contributions from each channel input of the first phase. Thus, after canceling the undesired message contributions, User~3 receives exactly the $m$ bits of $\msg_3$.
\end{description}

The key intuition on why we XOR uncoded bits of some messages with RLCs of others is as follows.
Assume our objective is to create an input signal such that User~1 can use side information to cancel ``interference'' from undesired messages,\footnote{We focus on User~1's ability to cancel contributions of other messages, since by assumption User~2 and 3 have full knowledge of undesired messages and can cancel any such interference perfectly.} $\msg_2$ and $\msg_3$. As $m$ grows large the probability that User~1 can cancel a random combination of $\msg_2$ or $\msg_3$ vanishes, and thus RLCs of $\msg_2$ and $\msg_3$ almost surely add to interference that cannot be canceled. However, by XORing uncoded $\msg_2$ and $\msg_3$ bits, the probability that User~1 can exploit side information to cancel interference remains constant regardless of $m$. 
It is worth noting that while Users~2 and 3 eventually know all three messages (through side information and decoding their desired messages), User~1 ends up knowing only parts of messages $\msg_2$ and $\msg_3$. 

\begin{remark}
To obtain further intuition, we can also interpret the proposed scheme as a form of interference alignment. 
Let $\msg_i^+$ and $\msg_i^-$ for $i=2,3$ denote subvectors of $\msg_i$ known and unknown respectively to User~1 via side information. 
Note that $\msg_i^-$ may be thought of as the part of $\msg_i$ that interferes with User~1 getting $\msg_1$, and that it cannot be canceled. Additionally, notice that the lengths of $\msg_i^+$ and $\msg_i^-$ are approximately $\frac{m}{3}$ and $\frac{2m}{3}$, respectively. 

First, consider what strategy the sender could use if it was not blind and could identify these subvectors. It could first send RLCs of $\msg_2^-$ and $\msg_3^-$ bits knowing that all such content cannot be cancelled by User~1. Then it could send RLCs of $\msg_1$, $\msg_2^+$, and $\msg_3^+$ bits, knowing that User~1 can cancel the $\msg_2$ and $\msg_3$ contribution for \emph{every} such input. Specifically, the non-blind sender \emph{aligns} the bits User~1 can cancel ($\msg_2^+$ and $\msg_3^+$), as well as the bits it cannot cancel ($\msg_2^-$ and $\msg_3^-$).  Via an equation counting argument, it is easy to verify that such a scheme achieves a higher rate of $\frac{3}{5}$

When the sender is blind, it is unable to distinguish $\msg_2^+$ from $\msg_2^-$ and $\msg_3^+$ from $\msg_3^-$. However, we would still like to efficiently send both $\msg_2$ and $\msg_3$ simultaneously, and therefore our scheme achieves such alignment \emph{probabilistically} by using uncoded bits from $\msg_2$ ($\msg_3$) to preserve the separation between $\msg_2^-$ and $\msg_2^+$ ($\msg_3^-$ and $\msg_3^+$). Figure~\ref{fig:XORexample2} highlights the two desired interference alignment cases, as well as the transmissions where alignment fails due to the sender being blind. \label{rem:alignment}
\end{remark}

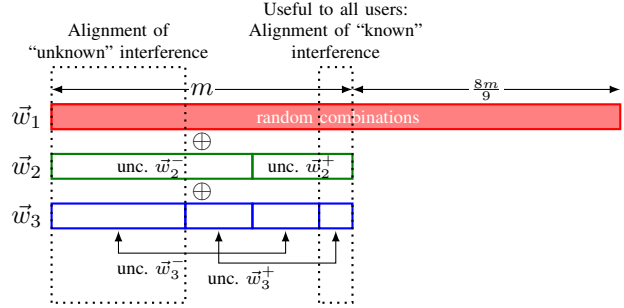
\begin{figure}\centering
\begin{tikzpicture}[yscale=0.66]
\draw [thick,blue] (0,0) rectangle (1.778,0.5);
\draw [thick,blue] (1.778,0) rectangle (2.67,0.5);
\draw [thick,blue] (2.67,0) rectangle (3.556,0.5);
\draw [thick,blue] (3.556,0) rectangle (4,0.5);
\draw [thick,green!50!black] (0,1) rectangle (2.67,1.5);
\draw [thick,green!50!black] (2.67,1) rectangle (4,1.5);
\draw [thick,red] (0,2) rectangle (4,2.5);
\draw (0,0.25) node[left] {$\msg_3$};
\draw (0,1.25) node[left] {$\msg_2$};
\draw (0,2.25) node[left] {$\msg_1$};
\draw (2,0.75) node {$\oplus$};
\draw (2,1.75) node {$\oplus$};
\draw[latex-latex] (0.889,-0.05) |- (0.889,-0.5) -- node[below,pos=0.2,inner sep=1pt] {\scriptsize unc. $\msg_3^-$}  (3.111,-0.5) -| (3.111,-0.05);
\draw[latex-latex] (2.222,-0.05) |- (2.222,-0.7) -- node[below,pos=0.2,inner sep=1pt] {\scriptsize unc. $\msg_3^+$}  (3.778,-0.7) -| (3.778,-0.05);
\draw (1.33,1.25) node {\scriptsize unc. $\msg_2^-$};
\draw (3.33,1.25) node {\scriptsize unc. $\msg_2^+$};
\filldraw [draw=red,fill=red!50,thick] (0,2) rectangle (7.56,2.5);
\draw (3.81,2.25) node[white] {\scriptsize random combinations};

\draw[latex-latex] (0,2.8) --node[fill=white,inner sep=0.5pt]{$m$} (4,2.8);
\draw[latex-latex] (4,2.8) --node[fill=white,inner sep=0.5pt]{\tiny $\frac{8m}{9}$} (7.56,2.8);
\draw [ thick,dotted] (3.556,-1.5) rectangle (4,3.25);
\draw [ thick,dotted] (0,-1.5) rectangle (1.778,3.25);
\draw (3.778,3.25) node[above]{\scriptsize \parbox[c]{10em}{\centering Useful to all users: Alignment of ``known'' interference}};
\draw (0.889,3.25) node[above]{\scriptsize \parbox[c]{10em}{\centering Alignment of ``unknown'' interference}};
\end{tikzpicture}\vspace{-2ex}
\caption{Illustration of the symmetric-capacity-achieving scheme of the example. The horizontal axis provides scale representation of the number of channel uses dedicated to each phase. Transmission type is illustrated using outlined (uncoded) or shaded (randomly coded) blocks. Notice that, because the sender is blind, parts of messages $\msg_2$ and $\msg_3$ that are known to User~1 cannot be explicitly \emph{aligned} as discussed in Remark~\ref{rem:alignment} and thus some parts of $\msg_3^-$ is XORed with $\msg_2^+$ as well as some of $\msg_2^-$ is XORed with $\msg_3^+$. These are displayed as contiguous blocks in the figure for clarity, but in reality would be interleaved throughout the first $m$ channel uses. \vspace{-5ex}}\label{fig:XORexample2}
\end{figure}

\begin{remark}
Our scheme's probabilisitic alignment is obviously less effective than the explicit alignment that is possible when the sender knows the side information. This loss of effectiveness is particularly well captured by the amount of additional interference incurred by our scheme, and thus any attempt at a converse must capture the amount of interference incurred \emph{as a result of blindness}. 

Notice that in both schemes, because we must send all of $\msg_2$ to User~2, we incur \emph{at least} $\frac{2m}{3}$ bits of interference from $\msg_2^-$. In the case of the non-blind sender, one can show that this is all the interference that is incurred at User~1, because the sender can fully align the interference caused from the messages $\msg_2^-$ and $\msg_3^-$ at User~1. 
On the other hand, because the sender is blind, our scheme incurs an additional $\frac{1}{3}\times\frac{2}{3}\times m$ bits of interference. In Section~\ref{sec:converse}, we present a general outer bound that shows that this additional interference is indeed unavoidable and hence the scheme presented above is indeed the best possible in this example. More specifically, we will prove a generalization of the following inequality:
\begin{align}
H(\vec{\mathbf{x}}|\msg_1,\msg_2^+,\msg_3^+) \geq{}& \frac{2}{3}H(\vec{\mathbf{x}}|\msg_1,\msg_3^+) + \frac{1}{3}H(\vec{\mathbf{x}} |\msg_1,\msg_2,\msg_3^+).
\label{eq:motexconv}
\end{align}

The above inequality lower bounds the interference at User~1 (i.e., the unknown parts of $\msg_2$ and $\msg_3$) with a convex combination of terms that either represent providing none of $\msg_2$ as side information ($H(\vec{\mathbf{x}}|\msg_1,\msg_3^+)$ or
all of $\msg_2$ as side information ($H(\vec{\mathbf{x}} |\msg_1\msg_2,\msg_3^+)$). The coefficient weights that describe the combination are a function of the side information parameter $\mu = \frac{2}{3}$. A more general form of this inequality is the key lemma used to construct the outer bound.\label{rem:exconv}
\end{remark}

Before concluding the section, we point out that this inequality is valid \emph{only when the sender is blind}. 
Indeed, if we consider the non-blind sender like in Remark~\ref{rem:alignment} the correct inequality would be 
\begin{align}
H(\vec{\mathbf{x}}|\msg_1,\msg_2^+,\msg_3^+) \geq{}& \frac{2}{3}H(\vec{\mathbf{x}}|\msg_1,\msg_3^+),\label{eq:motexconv2}
\end{align}
which is clearly looser than (\ref{eq:motexconv}). Note that the additional term that appears in (\ref{eq:motexconv}) but does not appear in (\ref{eq:motexconv2}) captures additional interference due to blindness of the server.

\section{Achievability}\label{sec:achieve}

In this section, we study achievable rates in the BIC problem. As alluded to in the motivating example, one possible approach to dealing with blind side information at users is random coding. For example, in a standard random linear code applied to binary message sequences, each channel input is created by XORing random linear combinations of \emph{all} bits from \emph{all} messages. In the rest of the paper, we refer to this approach as conventional random coding.
Using conventional random coding requires that all users decode all messages, thus rate tuples are achievable if and only if they satisfy, for every $i\in\{1,\ldots,K\}$,
\begin{align}
    r_i + \sum_{j\neq i} \mu_{ij}r_j \leq 1.\label{eq:randomnetworkrate}
\end{align}

In some cases, this suffices to achieve the full capacity region. For example, we will see that the rate region achievable by conventional random coding in the following two scenarios exactly matches the outer bounds derived in the next section:\footnote{Theorem~\ref{thm:KuOB} states the outer bound while the capacity regions for the two scenarios are formally stated as Propositions~\ref{prop:2u} and~\ref{prop:KuSym}, respectively.}
\begin{itemize}
    \item 2-user BICs, for any value of $\mu_{12}$ and $\mu_{21}$. 
    \item Symmetric $K$-user BICs, where $\mu_{ij}=\mu$ for all $i\neq j$. 
\end{itemize}
However, as demonstrated in the previous example, conventional random coding is not optimal in general, and in the rest of this section we propose a new hybrid encoding strategy that XORs random combinations of all bits from \emph{some} messages with \emph{uncoded bits} from others. This hybrid between random coding and uncoded transmission is the key mechanism to blindly exploit side information. For simplicity, we focus on symmetric rates achievable in an arbitrary 3-user BIC. We will first state the achievable symmetric rate as a theorem, then describe the encoding and decoding strategies before finally proving that the symmetric rate claimed in the theorem is indeed achievable.

\subsection{3-user BIC Hybrid Coding}

We now state the 3-user symmetric rate achievable using hybrid coding. We then define a hybrid encoding scheme for 3-user BIC problems, using key points from the motivating example.
\begin{theorem}\label{thm:3uACH}
Consider a 3-user BIC problem, defined by parameters $\{\mu_{ij}\}$, where WLOG\footnote{For any three users, such a condition must hold for at least one permutation of indices.} user indices are such that
\begin{align}
\mu_{32}\leq{}&\mu_{23}\leq \max\{\mu_{1i},\mu_{i1}\},\label{eq:choosing}
\end{align}
for either $i\in\{2,3\}$. Any $r_{sym}$ satisfying the following is achievable:
\begin{align}
r_{sym}^{}\leq{}&\min\left\{\frac{1}{1+\mu_{21}+\mu_{23}},\frac{1}{1+\mu_{31}+\mu_{32}}\right\},\label{eq:3uACH1}\\
r_{sym}^{}\leq{}&\max\bigg\{\frac{1}{1+\mu_{23}+\mu_{12} +\mu_{13}(1-\mu_{23}+\mu_{32})(1-\mu_{12})},
    \frac{1}{1+\mu_{12} +\mu_{13}}\bigg\}.\label{eq:3uACH2}
\end{align}
\end{theorem}

\begin{remark}
Consider $r_{sym}$ satisfying (\ref{eq:3uACH1}) and (\ref{eq:3uACH2}). In the right hand side of (\ref{eq:3uACH2}), if the second term within the maximization is larger, then (\ref{eq:3uACH1}) and (\ref{eq:3uACH2}) simplify to
$r_{sym}\leq \min\bigg\{\frac{1}{1+\mu_{21}+\mu_{23}},\frac{1}{1+\mu_{31}+\mu_{32}},
\frac{1}{1+\mu_{12}+\mu_{13}}\bigg\}$.
In this case, from (\ref{eq:randomnetworkrate}) it is clear that conventional random coding suffices to achieve the desired rate. 
Hence, our hybrid coding scheme increases the symmetric rate whenever the first term in the max of (\ref{eq:3uACH2}) is larger. Additionally, since conventional random coding suffices when the second term is larger, to prove Theorem~\ref{thm:3uACH}, we need only to describe a scheme and prove achievability for $r_{sym}$ satisfying
\begin{align}
r_{sym}\leq& \min\bigg\{\frac{1}{1+\mu_{21}+\mu_{23}},\frac{1}{1+\mu_{31}+\mu_{32}},\frac{1}{1+\mu_{23}+\mu_{12} +\mu_{13}(1-\mu_{23}+\mu_{32})(1-\mu_{12})}\bigg\}.\label{eq:3uACH2x}
\end{align}
\end{remark}

We now define our hybrid coding scheme where, for any $r_{sym}$ satisfying (\ref{eq:3uACH2x}), the sender will communicate $m = nr_{sym}-\delta_n$ bits where ($\delta_n$ is chosen such that $\delta_n=o(n)$) to each user in $n$ channel uses, such that probability of error vanishes as $n$ goes to infinity.\footnote{The $o(n)$ term 1) accounts for the fact that $m$ must be integer, and 2) as we shall see, ensures that decoding error will vanish as $n$ grows large.}

\subsubsection*{Encoding} 

The hybrid coding scheme is characterized by three parameters, $N_1$, $N_2$, and $N_3$. 
For each $i\in\{1,2,3\}$, we generate $N_i$ random linear combinations (RLC) of the bits only in $\msg_i$, denoted by vector $\vec{J}_i$. 
The precise values of $N_1$, $N_2$, and $N_3$ are specified later, however we point out as depicted in Figure~\ref{fig:3coding}, that $N_1-m\geq N_2\geq N_3$. 

As shown in the figure, the sender combines RLCs and uncoded bits of messages in five phases. During Phase~1, each input is the XOR of one bit from each of $\vec{J}_1$, $\vec{J}_2$, and $\vec{J}_3$, where we take bits from each vector sequentially. 
Phase~1 ends and Phase~2 begins when the bits in $\vec{J}_3$ are exhausted (i.e., after $N_3$ channel uses). 
Similarly, the number of channel uses allocated to each phase of transmission are dictated by when we exhaust the bits of a certain type: 
Phase~2 inputs consist of an XOR of $\vec{J}_1$, $\vec{J}_2$, and $\msg_3$ bits, and ends when we have no more bits from $\vec{J}_2$; 
Phase~3 inputs consist of an XOR of $\vec{J}_1$, $\msg_2$, and $\msg_3$ bits, and ends when we have no more bits from $\msg_3$; 
Phase~4 inputs consist of  an XOR of $\vec{J}_1$ and $\msg_2$ bits, and ends when we have no more bits from $\msg_2$; 
and Phase~5 inputs consist of only a $\vec{J}_1$ bit.

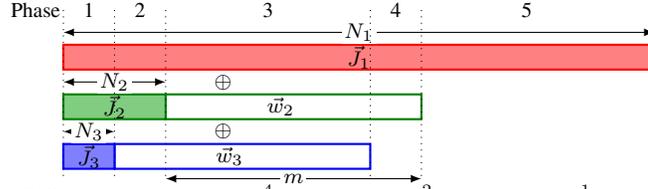
\begin{figure}
\centering
\begin{tikzpicture}[xscale=1.7,yscale=0.66,font=\footnotesize]

\draw [thick,blue] (0.4,0) rectangle (2.4,0.5);
\draw (1.3,0.25) node{\footnotesize $\msg_3$};
\filldraw [draw=blue,fill=blue!50,thick] (0,0) rectangle (0.4,0.5);
\draw[] (0.2,0.25) node[]{\footnotesize $\vec{J}_3$};

\draw [thick,green!50!black] (0.8,1) rectangle (2.8,1.5);
\draw (1.7,1.25) node{\footnotesize $\msg_2$};
\filldraw [draw=green!50!black,fill=green!60!black!50,thick] (0,1) rectangle (0.8,1.5);
\draw (0.4,1.25) node[]{\footnotesize $\vec{J}_2$};

\filldraw [draw=red,fill=red!50,thick] (0,2) rectangle (4.624,2.5);
\draw (2.312,2.25) node[]{\footnotesize $\vec{J}_1$};

\draw[latex-latex,thin] (0,2.75) --node [fill=white,inner sep=1pt]{$N_1$} (4.624,2.75);
\draw[latex-latex,thin] (0,1.75) --node [fill=white,inner sep=1pt]{$N_2$} (0.8,1.75);
\draw[latex-latex,thin] (0,0.75) --node [fill=white,inner sep=1pt]{$N_3$} (0.4,0.75);
\draw[latex-latex,thin] (0.8,-0.2) --node [fill=white,inner sep=1pt]{$m$} (2.8,-0.2);
\draw[dotted,thin] (0,3.2) -- (0,-0.2);
\draw (0.05,3.2) node[left] {Phase};
\draw (0.2,3.2) node {1};
\draw[dotted,thin] (0.4,3.2) -- (0.4,-0.2);
\draw (0.6,3.2) node {2};
\draw[dotted,thin] (0.8,3.2) -- (0.8,-0.2);
\draw (1.25,1.75) node {$\oplus$};
\draw (1.25,0.75) node {$\oplus$};
\draw (1.6,3.2) node {3};
\draw[dotted,thin] (2.4,3.2) -- (2.4,-0.2);
\draw (2.6,3.2) node {4};
\draw[dotted,thin] (2.8,3.2) -- (2.8,-0.2);
\draw (3.624,3.2) node {5};
\draw[dotted,thin] (4.624,3.2) -- (4.624,-0.2);
\end{tikzpicture}\vspace{-2.5ex}
\caption{Hybrid coding scheme for 3-user BIC, where $\mu_{12}=\mu_{13} = \frac{4}{5}$, $\mu_{21}=\mu_{23}=\frac{2}{5}$, $\mu_{31}=\mu_{32}=\frac{1}{5}$. 
Outlined boxes represent uncoded bits, shaded boxes represent RLCs of a single message.
\vspace{-4.5ex}}\label{fig:3coding}
\end{figure}

\subsubsection*{Decoding}
We now describe the decoding scheme of each user. Users~2 and 3 each decodes all 3 messages. As in conventional random coding, this requires that User~2 and 3 each receive a sufficient number of independent linear combinations of messages bits, either via side information or the shared channel. 

A key point in our coding scheme lies in how User~1 exploits the hybrid coding structure to decode $\msg_1$.
As in the example, User~1 uses side information to cancel out the combinations of known $\msg_2$ and $\msg_3$ bits from symbols received in Phases~3 and 4. It uses these ``clean'' RLC of only $\msg_1$ bits along with those RLC received during Phase~5 to linearly decode only $\msg_1$.

For the scheme to achieve $r_{sym}$ (i.e., in order for decoding error probability to vanish as $n$ grows large), we claim that choosing $N_1$, $N_2$, and $N_3$ as
\begin{align}
N_1 = n\quad\text{ and }\quad
N_2 = nr_{sym}\mu_{23}\quad\text{ and }\quad
N_3 = nr_{sym}\mu_{32}\label{eq:codelength},
\end{align}
results in a probability of decoding error that vanishes as $n\rightarrow\infty$. We prove this formally in the following subsection.

\begin{remark}
Recall from the illustrative example, we wanted to maximize the chance that User~1 can clean $\msg_2$ and $\msg_3$ content from a transmission, and we assumed that both User~2 and 3 decode all three messages. Thus the phases of transmission in Figure~\ref{fig:3coding}, have the following roles: 
Phase~5 provides RLCs about $\msg_1$ to User~1. Phase~5 also provides enough $\msg_1$ RLCs for each of User~2 and User~3 to decode $\msg_1$ (with the help of their side information). 
A fraction, $(1-\mu_{12})$, of Phase~4 is useful to User~1 after using side information to clean the $\msg_2$ component, to obtain a clean RLC of only $\msg_1$. Similarly, a smaller fraction, $(1-\mu_{12})(1-\mu_{13})$, of Phase~3 is useful to User~1 by cleaning both the $\msg_2$ and $\msg_3$ components, to obtain a clean RLC of $\msg_1$. Note that User~1 only uses clean RLCs from Phases~3-5 to decode $\msg_1$. Because Users~2 and 3 each decoded $\msg_1$ from Phase 5, each cancels out the $\msg_1$ component from Phases~1-4, and then each uses the remaining residual symbols to decode both messages $\msg_2$ and $\msg_3$. 
\end{remark}

\subsection{Proof of Achievability}\label{sec:BIC:achieve}

We now address the achievability of rate $r_{sym}$ satisfying (\ref{eq:3uACH2x}), using the hybrid network codes we just defined. Before proceeding we recall that the message size $m$ is such that $m=nr-\delta_n$, where $\delta_n$ is positive and $\delta_n=o(n)$.

To prove that the rate is achievable, we must show that the probability that any user does not decode its desired message vanishes as $n\rightarrow \infty$ (i.e., $\Pr[ \widehat{\msg}_i \neq \msg_i] \rightarrow 0$). Moreover, since our decoding strategy requires that User~2 and User~3 decode all three messages, we also show that the probability of decoding error of all messages at Users~2 and 3 vanishes as $n\rightarrow\infty$. Specifically, we have the following possible error events, each of which must approach 0 as $n\rightarrow\infty$:
\begin{description}
    \item[$\mathcal{E}_1$:] User~1 fails to decode $\msg_1$.
    \item[$\mathcal{E}_2$:] User~2 fails to decode $\{\msg_1,\msg_2,\msg_3\}$.
    \item[$\mathcal{E}_3$:] User~3 fails to decode $\{\msg_1,\msg_2,\msg_3\}$.  
\end{description}

For each event we will separate the error analysis into different sources of error, and for each source of error we will use one of two analysis techniques to prove that the probability of such an event occurring vanishes with large $n$. In order to provide clarity, and since User~1's decoding strategy was the primary difference between hybrid coding and conventional random coding, we will revisit these two techniques after first applying them in the context of analyzing the probability of event $\mathcal{E}_1$.

Recall that User~1 first uses its side information to ``clean'' transmissions from Phases~3 and 4 resulting in random linear combinations (RLCs) of only bits from $\msg_1$. It then combines clean RLCs with those received during Phase~5 (recall from Figure~\ref{fig:3coding} that Phase~5 only has $\msg_1$ content) and attempts to linearly decode $\msg_1$. Therefore, we express User~1's decoding error as the union of two events, $\mathcal{E}_1 = \mathcal{E}_{1a}\cup\mathcal{E}_{1b}$, defined as:
\begin{description}
\item[$\mathcal{E}_{1a}$: ] The total number of random linear combinations (RLCs) cleaned from Phases~3 and 4 and received in Phase~5 is less than $m+\underline{\delta}_n$, where $\underline{\delta}_n$ grows with $n$ and $0 < \underline{\delta}_n < \delta_n$.
\item[$\mathcal{E}_{1b}$: ] The random matrix that describes the transformation of $\msg_1$ to received (clean) RLCs is rank deficient.
\end{description}
We will now proceed to show
\begin{align*}
\Pr[\mathcal{E}_{1}] =\Pr[\mathcal{E}_{1a}\cup\mathcal{E}_{1b}] =\Pr[\mathcal{E}_{1a}] + \Pr[\mathcal{E}_{1a}^c\cap\mathcal{E}_{1b}] = o(n).
\end{align*}

We first address $\Pr[\mathcal{E}_{1a}]$. By the scheme's construction: 
\begin{itemize}
\item Phase~3 has duration $m-N_2+N_3$, and the probability of cleaning each transmission is $(1-\mu_{12})(1-\mu_{13})$. 
\item Phase~4 has duration $N_2-N_3$, and the probability of cleaning each transmission is $(1-\mu_{12})$.
\item Phase~5 has duration $N_1-N_2-m$, and each transmission is a clean RLC of $\msg_1$.
\end{itemize}

We may thus represent receiving a clean RLC in the $\ell$-th channel use of Phase 3 as a Bernoulli($1-\mu_{12}-\mu_{13}+\mu_{12}\mu_{13}$) random variable $\lambda_3[\ell]$ which is  i.i.d. across $\ell=1,\ldots,m-N_2+N_3$ and receiving a clean equation in the $\ell$-th channel use of Phase 4 as a Bernoulli($1-\mu_{12}$) random variable $\lambda_4[\ell]$ which is  i.i.d. across $\ell=1,\ldots,N_2-N_3$. 
We now note that the duration of Phases~3 and 4 ($D_3$ and $D_4$) are by construction,
\begin{align*}
D_3 ={}& m-N_2+N_3 ={} nr_{sym}(1-\mu_{23}+\mu_{32}) -\delta_n,\\
D_4={}& N_2-N_3 ={} nr_{sym}(\mu_{23}-\mu_{32}),
\end{align*}
and that the duration of Phase~5 may be bounded as
\begin{align}
D_5={}&N_1-N_2-m \nonumber\\
    ={}& n-nr_{sym}(1+\mu_{23}) +\delta_n \nonumber\\
    \stackrel{(a)}{\geq}{}& nr_{sym}(1+\mu_{23}+\mu_{12} +\mu_{13}(1-\mu_{23}+\mu_{32})(1-\mu_{12}))
        -nr_{sym}(1+\mu_{23})) +\delta_n \nonumber\\
    ={}& nr_{sym}(\mu_{12} +\mu_{13}(1-\mu_{23}+\mu_{32})(1-\mu_{12}))
        +\delta_n \nonumber\\
    ={}& nr_{sym}(1-(1-\mu_{23}+\mu_{32})(1-\mu_{12})(1-\mu_{13}))
        -(\mu_{23}-\mu_{32})(1-\mu_{12})
        +\delta_n \nonumber\\
    \geq{}& nr_{sym}-D_3(1-\mu_{12})(1-\mu_{13}))-D_4(1-\mu_{12}),\label{eq:achproofbnd}
\end{align}
where step (a) results directly from (\ref{eq:3uACH2x}). From this, the probability of $\mathcal{E}_{1a}$ occurring is, in the limit,
\begin{align*}
    \lim_{n\rightarrow\infty}\Pr[\mathcal{E}_{1a}] 
    ={}& \lim_{n\rightarrow\infty}
        \Pr\Bigg[\Bigg(\sum_{\ell=1}^{D_3}\lambda_3[\ell] +\sum_{\ell^\prime=1}^{D_4}\lambda_4[\ell^\prime]  + D_5\Bigg)<m+\underline{\delta}_n\Bigg]\\
    \stackrel{(b)}{\leq}{}& \lim_{n\rightarrow\infty}
        \Pr\Bigg[\Bigg(\sum_{\ell=1}^{D_3}\lambda_3[\ell] +\sum_{\ell^\prime=1}^{D_4}\lambda_4[\ell^\prime]  + nr_{sym}
        -D_3(1-\mu_{12})(1-\mu_{13})) -D_4(1-\mu_{12})\Bigg)<m+\underline{\delta}_n\Bigg]\\
    ={}& \lim_{n\rightarrow\infty}
        \Pr\Bigg[\Bigg(\sum_{\ell=1}^{D_3}\lambda_3[\ell]-\E\left[\sum_{\ell=1}^{D_3}\lambda_3[\ell]\right] 
        +\sum_{\ell^\prime=1}^{D_4}\lambda_4[\ell]-\E\left[\sum_{\ell=1}^{D_3}\lambda_3[\ell]\right] 
         +nr_{sym} 
        \Bigg)<m+\underline{\delta}_n\Bigg]\\
    ={}& \lim_{n\rightarrow\infty}
        \Pr\Bigg[\Bigg(\sum_{\ell=1}^{D_3}\lambda_3[\ell]-\E\left[\sum_{\ell=1}^{D_3}\lambda_3[\ell]\right] 
        + \sum_{\ell^\prime=1}^{D_4}\lambda_4[\ell]-\E\left[\sum_{\ell=1}^{D_3}\lambda_3[\ell]\right] + nr_{sym} 
        \Bigg)<nr_{sym}-\delta_n+\underline{\delta}_n\Bigg]\\
    \leq{}& \lim_{n\rightarrow\infty}
        \Pr\left[\left(\frac{\left|\sum_{\ell=1}^{D_3}\lambda_3[\ell]-\E\left[\sum_{\ell=1}^{D_3}\lambda_3[\ell]\right]\right|}{n} 
        + \frac{\left|\sum_{\ell^\prime=1}^{D_4}\lambda_4[\ell]-\E\left[\sum_{\ell=1}^{D_3}\lambda_3[\ell]\right]\right|}{n}        \right)>\frac{\delta_n-\underline{\delta}_n}{n}\right]\\
    \stackrel{(c)}{=}{}& 0,
\end{align*}
where in (b) we applied the bound (\ref{eq:achproofbnd}) while noting that if $a\geq b$ then $\Pr[a<c] \leq \Pr[b<c]$,
and in (c) we invoked the law of large numbers while noting that $\delta_n-\underline{\delta}_n$ is positive by construction.

Now consider the event $\mathcal{E}_{1a}^c\cap\mathcal{E}_{1b}$. This describes the case where User~1 receives enough (i.e., $m+\underline{\delta}_n$) clean equations but the randomly generated matrix that maps $\msg_1$ to clean equations has rank less than $m$. This type of error is well studied throughout the network coding literature. For instance, from expression~(3) of~\cite{MacKay2005}, the probability of a $m\times m+\underline{\delta}_n$ random binary matrix having rank less than $m$ can be bounded as
\begin{align*}
    \Pr(\mathcal{E}_{1a}^c\cap\mathcal{E}_{1b}) \leq 2^{-\underline{\delta}_n},
\end{align*}
which implies, as desired,
\begin{align*}
    \lim_{n\rightarrow\infty}\Pr(\mathcal{E}_{1a}^c\cup\mathcal{E}_{1b}) = 0.
\end{align*}

We now revisit the analysis and note that $\mathcal{E}_{1a}$ may be thought of as the event where the actual amount of randomly available side information was ``not enough'' because it deviated significantly from the mean. On the other hand $\mathcal{E}_{1a}^c\cap\mathcal{E}_{1b}$ describes the case where there was a sufficient amount of side information, but the randomly generated coding scheme failed to communicate the remaining desired message bits. For each type of error we applied a different analysis technique. To address the first, we applied a concentration inequality to show that the probability that the amount of randomly available side information deviates significantly from the mean vanishes as $n$ grows large. To address the second, we applied existing analysis on the properties of randomly generated matrices to show that the probability of a rank-deficient encoding matrix vanishes as $n$ grows large. 

To prove that the probabilities of error events $\mathcal{E}_2$ and $\mathcal{E}_3$ also vanish as $n$ grows large, we must systematically break down these error events into subevents of these two types. Since these two users apply the same decoding process, we now focus on User~2, and we identify such subevents.

Recall the User~2 first uses its side information and Phase~5 transmissions (i.e., RLCs with only $\msg_1$ content) to decode $\msg_1$. Let $\mathcal{E}_{2,1}$ be the event where User~2 fails to decode $\msg_1$ which we further breakdown into the following subevents, $\mathcal{E}_{2,1a}$ and $\mathcal{E}_{2,1b}$: 
\begin{description}
\item[$\mathcal{E}_{2,1a}$: ] User~2 does not receive enough side information, i.e., $\mathbf{1}^\top\vec{g}_{21} < \mu_{21}nr_{sym} +\underline{\delta}_n$, where $\underline{\delta}_n$ grows with $n$ and $0 < \underline{\delta}_n < \delta_n$.
\item[$\mathcal{E}_{2,1b}$: ] The random matrix that describes the transformation of $\msg_1$ to Phase~5 RLCs is rank deficient,
\end{description}
One can verify using the same methods as in the analysis of $\mathcal{E}_1$ that the probability of either $\mathcal{E}_{21a}$ or $\mathcal{E}_{21a}^c\cap\mathcal{E}_{21a}$ occurring vanishes with large $n$ as long as the rate $r_{sym}$ satisfies (\ref{eq:3uACH2x}).

Next, recall that after decoding $\msg_1$, User~2 removes $\msg_1$ content from Phases~1--4, and proceeds to decode both $\msg_2$ and $\msg_3$. Let $\mathcal{E}_{2,\{2,3\}}$ denote the event where User~2 fails to decode $\{\msg_2,\msg_3\}$, and we now study specifically $\mathcal{E}_{2,1}^c \cap \mathcal{E}_{2,\{2,3\}}$. 
Notice that by construction, after the $\msg_1$ content has been removed, User~2 will receive some uncoded bits of $\msg_2$ from Phase~4. Furthermore, User~2 can also use its side information to clean the $\msg_3$ component from some transmissions during Phase~3 to receive more (independent) uncoded bits from $\msg_2$. 
Noting this observation, we can now specify the final two error events for $\mathcal{E}_2$ analysis:
\begin{description}
\item[$\mathcal{E}_{2,1}^c \cap \mathcal{E}_{2,\{23\}a}$: ] \quad \quad \quad \quad After decoding $\msg_1$ and removing its content from Phases~1--4, the number of bits about $\msg_3$ User~2 learns from side information and the number of clean uncoded bits about $\msg_2$ User~2 learns from Phases~3 and 4 is significantly less than the mean.
\item[$\mathcal{E}_{2,1}^c \cap \mathcal{E}_{2,\{23\}b}$: ] \quad \quad \quad \quad After decoding $\msg_1$ and removing its content from Phases~1--4, the random matrix that describes the transformation of $\msg_2$ and $\msg_3$ to transmissions in Phases~1 and 2 is rank deficient.
\end{description}
Again, one can verify using the same methods as in the analyses of $\mathcal{E}_{1a}$ and $\mathcal{E}_{1b}$ that the probabilities of these events occurring vanish with large $n$ as long as the rate $r_{sym}$ satisfies (\ref{eq:3uACH2x}). Using the analyses of these subevents, we have
\begin{align*}
\lim_{n\rightarrow\infty}\Pr\left[\mathcal{E}_2\right] 
    ={}& \lim_{n\rightarrow\infty}\Pr\left[\mathcal{E}_{2,1}\right] + \Pr\left[\mathcal{E}_{2,1}^c \cap \mathcal{E}_{2,\{23\}}\right] \\
    ={}& \lim_{n\rightarrow\infty}\Pr\left[\mathcal{E}_{2,1a}\right] + \Pr\left[\mathcal{E}_{2,1a}^c\cap\mathcal{E}_{2,1b}\right] 
        + \Pr\left[\mathcal{E}_{2,1}^c \cap \mathcal{E}_{2,\{23\}a}\right] 
        + \Pr\left[\mathcal{E}_{2,1}^c \cap \mathcal{E}_{2,\{23\}a}^c\cap \mathcal{E}_{2,\{23\}b}\right] \\
    ={}& 0.
\end{align*} 

Through similarly identifying subevents of $\mathcal{E}_3$ we can also establish that $\Pr[\mathcal{E}_3]\rightarrow 0$ as $n\rightarrow\infty$. Therefore, the probability of decoding error at each user vanishes as $n$ grows large.\hfill\qed

\section{Outer Bound}\label{sec:converse}

In this section, we present an outer bound on the capacity region of the BIC problem. We will first state and prove the bound for the 3-user setting and remark on its implications. We then introduce a key lemma and prove the 3-user outer bound. Finally we state a general expression for an outer bound on the general $K$-user BIC capacity region. Its proof is relegated to Appendix~\ref{app:KuOB}.

\subsection{3-user Outer Bound}
We begin by stating the following result:
\begin{theorem}\label{thm:3uOB}
Consider a 3-user BIC problem. Rates $(r_1,r_2,r_3)$ are achievable only if, 
\begin{align}
r_i + 
\mu_{ij}r_j + 
\left(\mu_{ik}-\frac{\plus{\mu_{ij}-\mu_{kj}}\plus{\mu_{ik}-\mu_{jk}}}{1-\mu_{kj}}\right)r_k 
\leq{}& 1\label{eq:thm:3uOB},
\end{align}
for any $i\neq j\neq k\in\{1,2,3\}$ and $\plus{a} \triangleq \max\{a,0\}$. 
\end{theorem}
%
%
%
%

\begin{remark}
If the sender is not blind (i.e., the side information is known), our BIC problem can be converted to an analogous classic index coding problem with each user, $i$, desiring four different messages whose rates sum to by $r_i$ and whose proportion are determined by $\mu_{ji}$ and $\mu_{ki}$ for $i\neq j\neq k$.


For this resulting classic index coding problem, using the coding techniques of~\cite{ABKSW2013:isit}, it can be shown that rate tuples $(r_1,r_2,r_3)$ satisfying, for all $i\neq j\neq k \in\{1,2,3\}$,
\begin{align*}
r_i +\mu_{ij}r_j +\mu_{ik}\mu_{jk}r_k\leq 1
\end{align*}
are achievable. Notice that for some side information parameters (e.g., when $\mu_{23}=\mu_{32}=0$ and $\mu_1j>0$ for $j=2,3$), the rates achieved by a non-blind sender can be greater than the BIC outer bound, (\ref{eq:thm:3uOB}). The key difference in expressions is the third term on the left side of (\ref{eq:thm:3uOB}), which captures (at least partially) the capacity loss due to sender blindness.
\end{remark}

\begin{remark}
By evaluating Theorem~\ref{thm:3uOB} and comparing with the condition for achievability using conventional random coding (\ref{eq:randomnetworkrate}) we arrive at the following result:
\begin{proposition}\label{prop:2u}
Consider a 2-user BIC defined by parameters $\mu_{12}$ and $\mu_{21}$. The capacity region is the set of all rate pairs $(r_1,r_2)$ satisfying
\begin{align}
r_1 + \mu_{12}r_2 \leq{}& 1,\label{eq:2u1}\\
\mu_{21}r_1 + r_2 \leq{}& 1.\label{eq:2u2}
\end{align}
\end{proposition}
\begin{proof}
The converse results from Theorem~\ref{thm:3uOB} by letting $i,j\in\{1,2\}$, $k=3$ and fixing $r_3=0$, while achievability is a result of evaluation of (\ref{eq:randomnetworkrate}).
\end{proof}
\end{remark}

\subsection{Proof of Theorem~\ref{thm:3uOB}}\label{sec:3uOB}

To prove Theorem~\ref{thm:3uOB}, we start by stating and proving a key lemma:
\begin{lemma}\label{lem:split}
Consider a BIC problem with side information parameters $\{\mu_{ij}\}$.
Then, for any $(r_1,\ldots,r_K)$ scheme with block length $n$ and any random variable $V$ that is independent of $\msg_j$ and $\vec{g}_{ij}$ with $i\neq j$ (but may depend on other messages and channel parameters), we have 
\begin{align}
H\left(\vec{\mathbf{x}}\middle|\sinfo_{ij},\vec{g}_{ij},V\right) \geq{}& \mu_{ij}H\left(\vec{\mathbf{x}}\middle|V\right) + \left(1-\mu_{ij}\right)H\left(\vec{\mathbf{x}} \middle|\msg_j,V\right).\label{eq:lem:split:1}
\end{align}
Additionally, if $\mu_{kj}\leq\mu_{ij}$ where $i\neq j \neq k$, then
\begin{align} 
H\left(\vec{\mathbf{x}}\middle|\sinfo_{ij},\vec{g}_{ij},V\right) \geq{}& \frac{\mu_{ij} - \mu_{kj}}{1-\mu_{kj}}H\left(\vec{\mathbf{x}}\middle|V\right) 
    + \frac{1-\mu_{ij}}{1-\mu_{kj}}H\left(\vec{\mathbf{x}} \middle|\sinfo_{kj},\vec{g}_{kj},V\right).\label{eq:lem:split:2}
\end{align}
\end{lemma}
\begin{remark}
Inequality (\ref{eq:lem:split:1}) captures an intuition that can be illustrated through the following toy problem. Consider a scenario where the sender has 4 bits $b_1,b_2,c_1,c_2$. It knows that User~2 knows $c_1$ and $c_2$ already and User~3 knows $b_1$ and $b_2$. On the other hand, the sender only knows that User~1 knows \emph{either} $b_1$ or $b_2$ (but not both) and \emph{either} $c_1$ or $c_2$ (but not both) and that both of these uncertainties are the result of a (fair) coin flip. If the sender sends a single transmission such that both User~2 and User~3 learn something new about $b_1,b_2,c_1,c_2$, what is the minimum probability that User~1 also learns something new?

One possible transmission would be to send $b_1\oplus c_1$. In this case, User~2 learns $c_1$ and User~3 learns $b_1$, and there is a 75\% chance that User~1 learns either $b_1$, $c_1$, or $b_1\oplus c_1$. In comparison, we can evaluate (\ref{eq:lem:split:1}) for the porposed transmission by letting $i=1$, $j=2$, $k=3$, $\mu_{12}=\mu_{13} = \frac{1}{2}$, and assuming $\msg_{2}=[b_1\quad b_2]$, $\msg_3 = [c_1\quad c_2]$, and $V = (\sinfo_{13},\vec{g}_{13})$.
In doing so, we see that the right hand side of (\ref{eq:lem:split:1}) evaluates to $ \mu_{12}(1) + (1-\mu_{12})(\mu_{13})= \frac{3}{4}$, signifying that the 75\% chance of ``leaking'' information to User~1 is the lowest possible. 

Notice that, as stated, Lemma~\ref{lem:split} does not assume decodability of any message. Moreover, it applies regardless of the number of channel uses, whereas the toy example assumed only a single channel use. Consequently, Lemma~\ref{lem:split} can be viewed as a powerful extension of the intution from the toy example to vector (i.e., coded) representations of message bits.
\label{rem:toyexample}\end{remark}
\begin{remark}\label{rem:alignment2}
One can note that if the transmitter is not blind, the sender can construct a signal that invalidates 4. For example, consider the scenario in Remark~\ref{rem:toyexample}, but now assume that the sender is aware that User~1 knows $b_1$ and $c_1$. The sender can now use this knowledge to send a single transmission $b_1\oplus c_1$. One can easily verify that, for this one transmission, the left hand side now evaluates to $0$, while the right hand side evaluates to $\frac{1}{2}$ which violates the claim. Therefore, the inequality specifically captures the impact of a blind sender.
\end{remark}
\begin{remark}
Inequality (\ref{eq:lem:split:1}) can more generally be interpreted as follows. Note that $H(\vec{x}|V)$ corresponds to the case that there is no side information about $\msg_j$ provided in the conditioning, and $H(\vec{x}|V,\msg_j)$ corresponds to the case that all of $\msg_j$ is provided as the side information in the conditioning. Therefore inequality (\ref{eq:lem:split:1}) lower bounds $H(\vec{\mathbf{x}}|\sinfo_{ij},\vec{g}_{ij},V)$ with a weighted average of two extreme cases, where either none or all of $\msg_j$ is provided as side information. A similar interpretation holds for (\ref{eq:lem:split:2}), where $\msg_j$ is replaced with $(\sinfo_{kj},\vec{g}_{kj})$.
\end{remark}

\begin{proof}
To prove Lemma~\ref{lem:split}, we first define a virtual side information signal, $\sinfo^\prime$, such that $\sinfo_{ij}$ is a physically degraded version of $\sinfo^\prime$. To do so we also specify two channel state sequences, $\vec{g}^\prime$ and $\vec{g}^{\ddagger}$ drawn i.i.d from two different Bernoulli distributions that take a values of zero with probabilities $\mu^\prime$ and $\delta = \frac{\mu_{ij}-\mu^\prime}{1-\mu^\prime}$, respectively. The side information signals are constructed such that for $\ell\in\{1,\ldots,m_j\}$,
\begin{align}
    \phi^\prime[\ell] ={}& g^\prime[\ell]{w}_j[\ell] \qquad \text{ and }\qquad
    g_{ij}[\ell] = g^\prime[\ell]g^\ddagger[\ell],
\end{align}
which necessarily implies $\mu^\prime \leq\mu_{ij}$. 

We now establish a relationship between the virtual side information signal, $\sinfo^\prime$, and the degraded side information, $\sinfo_{ij}$, using a strong data processing inequality proven in~\cite{AGKN2014:isit} which states, for random variables $U \leftrightarrow X \leftrightarrow Y$, that form a Markov chain,
\begin{align*}
I(Y;U) \leq s^*(X;Y)I(X;U),
\end{align*} 
where 
\begin{align*}
s^*(X;Y) \triangleq{}& \sup_{Q_X\neq P_X}\frac{D(Q_Y||P_Y)}{D(Q_X||P_X)},
\end{align*}
and $Q_Y$ is the marginal distribution of $Y$ from the joint distribution $Q_{XY} = P_{Y|X}Q_X$. 
We apply the strong data processing inequality by letting $U = (\vec{\mathbf{x}},V)$, $X = (\sinfo^\prime,\vec{g}^\prime)$, and $Y = (\sinfo_{ij},\vec{g}_{ij})$, to show
\begin{align}
H\left(\vec{\mathbf{x}}\middle|\sinfo_{ij},\vec{g}_{ij},V\right) 
={}& -I\left(\sinfo_{ij},\vec{g}_{ij};\vec{\mathbf{x}} \middle| V\right)+H\left(\vec{\mathbf{x}} \middle| V \right)\nonumber\\
={}& -I\left(\sinfo_{ij},\vec{g}_{ij};\vec{\mathbf{x}},V\right)+H\left(\vec{\mathbf{x}} \middle|V\right)\nonumber\\
    \geq{}& -s^*\left(\left(\sinfo^\prime,\vec{g}^\prime\right);\left(\sinfo_{ij},\vec{g}_{ij}\right)\right) 
          I\left(\sinfo^\prime,\vec{g}^\prime;\vec{\mathbf{x}},V\right) + H\left(\vec{\mathbf{x}}\middle|V\right)\nonumber\\
    \stackrel{(a)}{=}{}& -\frac{1-\mu_{ij}}{1-\mu^\prime}
        \left(I\left(\sinfo^\prime,\vec{g}^\prime;\vec{\mathbf{x}}\middle|V\right) + I\left(\sinfo^\prime,\vec{g}^\prime;V\right)\right)+H\left(\vec{\mathbf{x}}\middle|V\right)\nonumber\\
    ={}& -\frac{1-\mu_{ij}}{1-\mu^\prime}
        \left(H\left(\vec{\mathbf{x}}\middle|V\right)-H\left(\vec{\mathbf{x}}\middle|V,\sinfo^\prime,\vec{g}^\prime\right)\right)+H\left(\vec{\mathbf{x}}\middle|V\right)\nonumber\\
    ={}& \frac{\mu_{ij}-\mu^\prime}{1-\mu^\prime} H\left(\vec{\mathbf{x}} \middle|V\right) + \frac{1-\mu_{ij}}{1-\mu^\prime}H\left(\vec{\mathbf{x}}\middle|V,\sinfo^\prime,\vec{g}^\prime\right).\label{eq:lemproof1}
\end{align}
Step (a), where we evaluated $s^*\left(\left(\sinfo^\prime,\vec{g}^\prime\right);\left(\sinfo_{ij},\vec{g}_{ij}\right)\right)$, is proven in Appendix~\ref{app:sstar}. 
Recall that we only require $\mu^\prime<\mu_{ij}$ in order for the virtual signal to be properly defined, and we notice the following to complete the proof:
\begin{itemize}
\item If $\mu^\prime=0$, then $(\sinfo^\prime,\vec{g}^\prime) = (\msg_j,\vec{1})$ and we prove (\ref{eq:lem:split:1}),
\item If $\mu^\prime=\mu_{kj}<\mu_{ij}$, then $(\sinfo^\prime,\vec{g}^\prime)$ is statistically equivalent to $(\sinfo_{kj},\vec{g}_{kj})$ and we prove (\ref{eq:lem:split:2}).\vspace{-3ex}
\end{itemize}\end{proof}

We now use Lemma~\ref{lem:split} to prove Theorem~\ref{thm:3uOB}. First, we note that two side info parameter relationships affect the form of (\ref{eq:thm:3uOB}): The term $\frac{\plus{\mu_{ij}-\mu_{kj}}\plus{\mu_{ik}-\mu_{jk}}}{1-\mu_{kj}}$ is nonzero only if both $\mu_{kj} < \mu_{ij}$ and $\mu_{jk} < \mu_{ik}$. In this case,
\begin{align}
r_i + 
\mu_{ij}r_j + 
\left(\mu_{ik}-\frac{(\mu_{ij}-\mu_{kj})(\mu_{ik}-\mu_{jk})}{1-\mu_{kj}}\right)r_k 
\leq{}& 1.\label{eq:3u2}
\end{align}
Otherwise, if either $\mu_{kj} \geq \mu_{ij}$ or  $\mu_{jk} \geq \mu_{ik}$, then
\begin{align}
r_i + 
    \mu_{ij}r_j + 
    \mu_{ik}r_k 
\leq{}& 1.\label{eq:3u1}
\end{align}
We prove these two cases separately, and only address the first case, (\ref{eq:3u2}), here. The proof of (\ref{eq:3u1}) will use similar techniques, and may be found in Appendix~\ref{app:3u2}. We therefore assume $\mu_{kj}< \mu_{ij}$ and $\mu_{jk}< \mu_{ik}$, and start with Fano's inequality at User~$i$:
\begin{align}
    nr_i 
    \leq{}& I\left(\vec{x},\sinfo_{ij},\vec{g}_{ij},\sinfo_{ik},\vec{g}_{ik};\msg_i\right) + o(n)\nonumber\\
    ={}& H\left(\vec{x}\middle|\sinfo_{ij},\vec{g}_{ij},\sinfo_{ik},\vec{g}_{ik}\right) - H\left(\vec{x}\middle|\msg_i,\sinfo_{ij},\vec{g}_{ij},\sinfo_{ik},\vec{g}_{ik}\right) + o(n)\nonumber\\
    \leq{}& n - H\left(\vec{x}\middle|\msg_i,\sinfo_{ij},\vec{g}_{ij},\sinfo_{ik},\vec{g}_{ik}\right) + o(n)\label{eq:3uprooffano}\\
    \stackrel{(a)}{\leq}{}& n - \frac{\mu_{ij} - \mu_{kj}}{1-\mu_{kj}}\overbrace{H\left(\vec{\mathbf{x}}\middle|\msg_i,\sinfo_{ik},\vec{g}_{ik}\right)}^A 
     - \frac{1-\mu_{ij}}{1-\mu_{kj}}\overbrace{H\left(\vec{\mathbf{x}} \middle|\msg_i,\sinfo_{kj},\vec{g}_{kj},\sinfo_{ik},\vec{g}_{ik}\right)}^B,\label{eq:3uprooffano2}
\end{align}
where in step (a) we applied (\ref{eq:lem:split:2}) from Lemma~\ref{lem:split} by letting $V=\left(\msg_i,\sinfo_{ik},\vec{g}_{ik}\right)$. Notice there are two negative entropy terms, $A$ and $B$, to account for. To address the quantity $A$, we enhance side information at User~$j$ from $\left(\sinfo_{ji},\vec{g}_{ji}\right)$ to $\msg_i$, and observe:
\begin{align}
    nr_j 
    \leq{}& I\left(\vec{x},\sinfo_{ji},\vec{g}_{ji},\sinfo_{jk},\vec{g}_{jk};\msg_j\right) +o(n)\nonumber\\
    \leq{}& I\left(\vec{x},\msg_{i},\sinfo_{jk},\vec{g}_{jk};\msg_j\right) +o(n)\nonumber\\
    ={}& H\left(\vec{x}\middle|\msg_{i},\sinfo_{jk},\vec{g}_{jk}\right) - H\left(\vec{x}\middle|\msg_{i},\msg_j,\sinfo_{jk},\vec{g}_{jk}\right) + o(n)\nonumber\\
    \stackrel{(b)}{\leq}{}& H\left(\vec{x}\middle|\msg_{i},\sinfo_{jk},\vec{g}_{jk}\right) - \mu_{jk}H\left(\vec{x}\middle|\msg_{i},\msg_j\right) + o(n)\nonumber\\
    \stackrel{(c)}{\leq}{}& \overbrace{H\left(\vec{x}\middle|\msg_{i},\sinfo_{ik},\vec{g}_{ik}\right)}^{A} - \mu_{jk}H\left(\vec{x}\middle|\msg_{i},\msg_j\right) + o(n),\label{eq:3uproof2A1}
\end{align}
where in (b) we used (\ref{eq:lem:split:1}) from Lemma~\ref{lem:split} while letting $V=(\msg_i,\msg_j)$, and 
in (c) we observe that, since $\mu_{kj} < \mu_{ij}$ and that $g_{jk}[\ell]=0$ implies $(\phi_{jk}[\ell]=0,g_{jk}[\ell]=0)$ is independent of $\msg_i$ and $\vec{x}$, replacing $\left(\sinfo_{jk},\vec{g}_{jk}\right)$ with $\left(\sinfo_{ik},\vec{g}_{ik}\right)$ reduces the effective conditioning (see Claim~\ref{cl:staten} in Appendix~\ref{app:KuOB}). At User~$k$ we enhance side information from $\left(\sinfo_{ki},\vec{g}_{ki},\sinfo_{kj},\vec{g}_{kj}\right)$ to $(\msg_j,\msg_k)$ to find:
\begin{align}
    nr_k 
    \leq{}& I\left(\vec{x},\sinfo_{ki},\vec{g}_{ki},\sinfo_{kj},\vec{g}_{kj};\msg_k\right) +o(n)\nonumber\\
    \leq{}& I\left(\vec{x},\msg_{i},\msg_j;\msg_k\right) +o(n)\nonumber\\
    \leq{}& H\left(\vec{x}\middle|\msg_{i},\msg_j\right) + o(n).\label{eq:3uproof2A2}
\end{align}

To account for the quantity $B$, we observe 
\begin{align}
n\mu_{ik}r_k 
    ={}& nr_k - n(1-\mu_{ik})r_k \nonumber\\
    \leq{}& I\left(\vec{x},\sinfo_{ki},\vec{g}_{ki},\sinfo_{kj},\vec{g}_{kj};\msg_k\right) - n(1-\mu_{ik})r_k + o(n)\nonumber\\
    \leq{}& I\left(\vec{x},\msg_{i},\sinfo_{kj},\vec{g}_{kj},\sinfo_{ik},\vec{g}_{ik};\msg_k\right) - n(1-\mu_{ik})r_k +o(n)\nonumber\\
    ={}& H\left(\vec{x}|\msg_{i},\sinfo_{kj},\vec{g}_{kj},\sinfo_{ik},\vec{g}_{ik}\right)  - H\left(\vec{x}\middle|\msg_{i},\sinfo_{kj},\vec{g}_{kj},\msg_k\right) 
     + I\left(\sinfo_{ik},\vec{g}_{ik};\msg_k\right) - n(1-\mu_{ik})r_k + o(n)\nonumber\\ 
    ={}& H\left(\vec{x}\middle|\msg_{i},\sinfo_{kj},\vec{g}_{kj},\sinfo_{ik},\vec{g}_{ik}\right)  - H\left(\vec{x}\middle|\msg_{i},\sinfo_{kj},\vec{g}_{kj},\msg_k\right) + o(n)\nonumber\\ 
    \stackrel{(d)}{\leq}{}& H\left(\vec{x}\middle|\msg_{i},\sinfo_{kj},\vec{g}_{kj},\sinfo_{ik},\vec{g}_{ik}\right)  - \mu_{ij}H\left(\vec{x}\middle|\msg_{i},\msg_k\right)+ o(n),\nonumber\\
        \leq{}& \overbrace{H\left(\vec{x}\middle|\msg_{i},\sinfo_{kj},\vec{g}_{kj},\sinfo_{ik},\vec{g}_{ik}\right)}^{B} - \mu_{kj}H\left(\vec{x}\middle|\msg_{i},\msg_k\right) + o(n).\label{eq:3uproof2B1}
\end{align}
In step (d) we used (\ref{eq:lem:split:2}). Also, like (\ref{eq:3uproof2A2}), we find
\begin{align}
    nr_j 
    \leq{}& H(\vec{x}|\msg_{i},\msg_j) + o(n).\label{eq:3uproof2B2}
\end{align}

By appropriately scaling (\ref{eq:3uproof2A1}), (\ref{eq:3uproof2A2}), (\ref{eq:3uproof2B1}), and (\ref{eq:3uproof2B2}), and then summing with (\ref{eq:3uprooffano2}) we arrive at (\ref{eq:3u2}), as desired.\hfill\qed

\subsection{$K$-user Outer Bound}
The construction of the bound is governed by a recursion specified using a tree data structure which we refer to as an \emph{outer bound tree}:
\begin{definition}[Outer Bound Tree (OBT)]\label{def:obt}
A $K$-user OBT is directed labeled tree with $K$ levels where each node in the first $K-2$ levels has 2 children and each node in the $K-1$-th level has one child. The label of the $i$-th node in level $\ell$ is denoted as $v[\ell,i]\in\{1,\ldots,K\}$, where if $\ell<K$ then $i\in \{1,\ldots,2^{\ell-1}\}$ and if $\ell=K$ then $i\in \{1,\ldots,2^{\ell-2}\}$. The index $i$ specifies the precise location in the level: nodes $i=2j-1$ and $i=2j$ in level $\ell<K$ are the left and right children, respectively, of a node $j$ in level $\ell-1$. Node $i$ in level $K$ is the sole child of node $i$ in level $K-1$. Finally, the labels of an OBT must satisfy the following:
\begin{enumerate}
    \item For any path from the root node of the tree to any leaf node, no labels are repeated.
    \item Any two nodes with the same parent cannot have the same label.
\end{enumerate}
\end{definition}
The first requirement is equivalent to saying that the sequence of labels along any path from root to leaf is a permutation of $\{1,\ldots,K\}$. This is demonstrated in Figure~\ref{fig:tree4}, where we provide an example of an 4-user OBT.

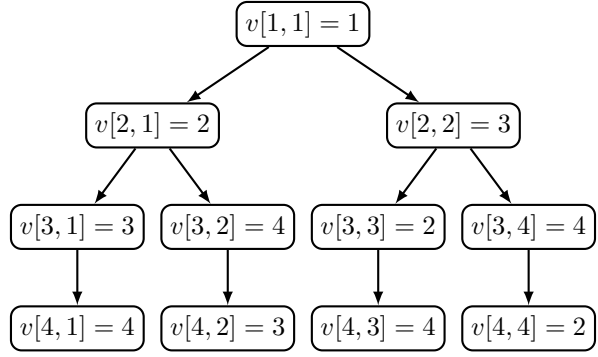
\begin{figure}
\centering
\begin{tikzpicture}[yscale=1.35]
\node (a) at (0,0) [draw,thick,rounded corners] {$v[1,1] = 1$};
\node (b) at (-2,-1) [draw,thick,rounded corners] {$v[2,1] = 2$};
\node (c) at (2,-1) [draw,thick,rounded corners] {$v[2,2] = 3$};
\node (d) at (-3,-2) [draw,thick,rounded corners] {$v[3,1] = 3$};
\node (e) at (-1,-2) [draw,thick,rounded corners] {$v[3,2] = 4$};
\node (f) at (1,-2) [draw,thick,rounded corners] {$v[3,3] = 2$};
\node (g) at (3,-2) [draw,thick,rounded corners] {$v[3,4] = 4$};
\node (h) at (-3,-3) [draw,thick,rounded corners] {$v[4,1] = 4$};
\node (i) at (-1,-3) [draw,thick,rounded corners] {$v[4,2] = 3$};
\node (j) at (1,-3) [draw,thick,rounded corners] {$v[4,3] = 4$};
\node (k) at (3,-3) [draw,thick,rounded corners] {$v[4,4] = 2$};
\draw[thick,-latex] (a)--(b);
\draw[thick,-latex] (a)--(c);
\draw[thick,-latex] (b)--(d);
\draw[thick,-latex] (b)--(e);
\draw[thick,-latex] (c)--(f);
\draw[thick,-latex] (c)--(g);
\draw[thick,-latex] (d)--(h);
\draw[thick,-latex] (e)--(i);
\draw[thick,-latex] (f)--(j);
\draw[thick,-latex] (g)--(k);
\end{tikzpicture}
\caption{Possible OBT for $K=4$ users. Notice that the sequence of labels along each root-to-leaf path is a permutation of the user indices.}\label{fig:tree4}\end{figure}

We now state the following outer bound for the $K$-user BIC:
\begin{theorem}[]\label{thm:KuOB} Consider a $K$-user BIC with $K\geq 3$, defined by parameters $\{\mu_{ij}\}$. The rate tuple $(r_1,\ldots,r_K)$ is achievable only if it satisfies, 
\begin{align}
\Gamma_A[1,1] \leq 1,
\end{align}
for any $K$-user OBT, where
\begin{align}
\Gamma_A[\ell,i] 
&=
\begin{cases}
r_{v[\ell,i]} + \zeta[\ell,i]\Gamma_A[\ell+1,2i-1] + (1-\zeta[\ell,i])\Gamma_B[\ell+1,2i] &\text{ if }\ell<K-1 \cr
r_{v[\ell,i]} + \zeta[\ell,i]r_{v[\ell+1,i]} &\text{ if }\ell=K-1 \cr
0 &\text{ otherwise } \cr
\end{cases},\label{eq:KuOBrec1}\\
\Gamma_B[\ell,i]
&=
\begin{cases}
\eta_{v[\ell,i]}\left[\ell-1,\left\lceil\frac{i}{2}\right\rceil\right]r_{v[\ell,i]} 
    + \zeta[\ell,i]\Gamma_A[\ell+1,2i-1] 
    + (1-\zeta[\ell,i])\Gamma_B[\ell+1,2i] &\text{ if }\ell<K-1 \cr
\eta_{v[\ell,i]}\left[\ell-1,\left\lceil\frac{i}{2}\right\rceil\right]r_{v[\ell,i]}
    + \zeta[\ell,i]r_{v[\ell+1,i]} &\text{ if }\ell=K-1 \cr
0 &\text{ otherwise } \cr
\end{cases},\label{eq:KuOBrec2}\\
\zeta[\ell,i] 
&=
\begin{cases}
\frac{\displaystyle \plus{\eta_{v[\ell+1,2i-1]}[\ell,i]-\mu_{v[\ell+1,2i],v[\ell+1,2i-1]}}}
{\displaystyle 1-\mu_{v[\ell+1,2i],v[\ell+1,2i-1]}} &\text{ if }\ell<K-1 \cr
 \eta_{v[\ell+1,2i-1]}[\ell,i]&\text{ if }\ell=K-1 \cr
0 &\text{ otherwise } \cr
\end{cases},\label{eq:KuOBrec3}
\end{align}
where we have, 
if $\ell>1$,
\begin{align}
\eta_j[\ell,i] 
&=
\begin{cases}
1 &\text{ if } j=v[\ell,i],\ i\text{ is odd}\cr
0 &\text{ if } j=v\left[\ell-1,\left\lceil\frac{i}{2}\right\rceil\right]\cr
\min\left\{\mu_{v[\ell,i],j},\eta_j\left[\ell-1,\left\lceil\frac{i}{2}\right\rceil\right]\right\} & \text{ otherwise }
\end{cases},\label{eq:KuOBrec4a}
\end{align}
and if $\ell=1$
\begin{align}
    \eta_j[1,1] = 
    \begin{cases}
    1 & \text{ if } j=v[1,1]\cr
    \mu_{v[1,1],j} & \text{ otherwise }
    \end{cases}. \label{eq:KuOBrec4}
\end{align}
\end{theorem}
The proof of Theorem~\ref{thm:KuOB} may be found in Appendix~\ref{app:KuOB}. Here we remark on how the intuitions from Theorem~\ref{thm:3uOB} are extended to $K$ users.

\begin{remark}
Consider the 3-user bound with respect to the more general statement of Theorem~\ref{thm:KuOB}. Figure~\ref{fig:tree3} depicts the exact assignment of labels for a 3-user OBT that results in Theorem~\ref{thm:3uOB}.
\begin{figure}
\centering
\begin{tikzpicture}[yscale=1.35]
\node (a) at (0,0) [draw,thick,rounded corners] {$v[1,1] = i$};
\node (b) at (-1,-1) [draw,thick,rounded corners] {$v[2,1] = j$};
\node (c) at (1,-1) [draw,thick,rounded corners] {$v[2,2] = k$};
\node (d) at (-1,-2) [draw,thick,rounded corners] {$v[3,1] = k$};
\node (e) at (1,-2) [draw,thick,rounded corners] {$v[3,2] = j$};
\draw[thick,-latex] (a)--(b);
\draw[thick,-latex] (a)--(c);
\draw[thick,-latex] (b)--(d);
\draw[thick,-latex] (c)--(e);
\end{tikzpicture}
\caption{The OBT for Theorem~\ref{thm:3uOB}.}\label{fig:tree3}\end{figure}
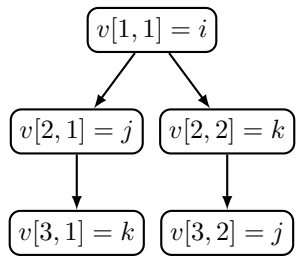

Recall that the construction of the outer bound in Theorem~\ref{thm:3uOB} began with applying Fano's inequality at User $i$ (i.e., the root node label of the OBT), and then applying (\ref{eq:lem:split:2}) of Lemma~\ref{lem:split}. Applying (\ref{eq:lem:split:2}) resulted in two terms $A$ and $B$ in (\ref{eq:3uprooffano2}), each of which was canceled by analysis of a different user with enhanced side information. This is reflected in the first case of (\ref{eq:KuOBrec1}), where in addition to the rate of the user associated with the node label, we have the quantities $\Gamma_A[\cdot]$ and $\Gamma_B[\cdot]$ associated with the expressions that will cancel $A$ and $B$, respectively. The scaling terms $\zeta[\ell,i]$ reflect the appropriate scaling terms needed for the cancellation; e.g., consider the final step in the proof of Theorem~\ref{thm:3uOB} where we took a weighted sum of (\ref{eq:3uprooffano2})--(\ref{eq:3uproof2B2}). The last quantity, $\eta_j[\ell,i]$, tracks the side information enhancement through each level of recursion.
\end{remark}

\begin{remark}
It is worth noting that the terms associated with the $K-1$-th layer of the OBT are special: This layer represents the ``base case'' of the recursion, and in the 3-user scenario, we reached this base case after only one application of (\ref{eq:lem:split:2}). At the $K-1$-th layer, instead of (\ref{eq:lem:split:2}) we apply (\ref{eq:lem:split:1}) which is reflected by the associated value of $\zeta[\ell,i]$ in (\ref{eq:KuOBrec3}).
\end{remark}

\begin{remark}
By evaluating Theorem~\ref{thm:KuOB} and comparing with the condition for achievability using conventional random coding (\ref{eq:randomnetworkrate}) we arrive at the following result:
\begin{proposition}\label{prop:KuSym}
Consider a $K$-user BIC where $\mu_{ij}=\mu$ for all $i\neq j$. The capacity region is the set of all rate tuples $(r_1,\ldots,r_K)$ satisfying for every $i\in\{1,\ldots,K\}$
\begin{align}
r_i + \mu\sum_{j\neq i} r_j \leq{}& 1\label{eq:KuSym}.
\end{align}
\end{proposition}
\begin{proof}
Achievability results directly from evaluation of (\ref{eq:randomnetworkrate}). To prove the converse, we observe that when $\mu_{ij}=\mu$ for all $j\neq i$, for all $\ell$
\begin{align}
\zeta[\ell,i] = 0,
\end{align}
and if $\ell>1$,
\begin{align*}
\eta_j[\ell,i]
=
\begin{cases}
1 &\text{ if } j=v[\ell,i]\text{ and $i$ is odd}\cr
0 &\text{ if } j=v[\ell-1,\left\lceil\frac{i}{2}\right\rceil]\cr
\mu & \text{ otherwise }
\end{cases}.
\end{align*}
Evaluating recursively through the OBT yields
\begin{align*}
    \Gamma_A[1,1] = r_{v[1,1]} + \mu r_{v[2,2]} + \ldots \mu r_{v[K,2^K]} \leq{}1.
\end{align*}
Since the path from root to leaf is a permutation of user indices (i.e., all user indices are represented and there exist no repeats), we arrive at (\ref{eq:KuSym}).
\end{proof}
\end{remark}

\section{Numerical Results}\label{sec:num}

In this section we perform numerical analysis of inner and outer bounds to illustrate 1) the gain in achievable rate of hybrid coding over conventional random coding, and 2) the gap between our derived inner and outer bounds.

To limit the scope of possible configurations  (parameterized by $\mu_{ij}$ terms), we focus on two symmetric scenarios for a representative set of parameters. In the first scenario, we consider side information that is ``one-sided symmetric'' (i.e., network parameters such that $\mu_{ij}=\mu_{ik}$ for all $i\neq j \neq k$) while in the second, we consider side information that is ``pairwise symmetric'' (i.e., network parameters such that $\mu_{ij}=\mu_{ji}$ for all $i\neq j$). For each scenario, we will assume that, the size of side information at User~1 is the least and at User~3 is the most, and we plot the following:
\begin{enumerate}
\item The 3-user outer bound of Theorem~\ref{thm:3uOB}, applied to the symmetric rate.
\item The achieved symmetric rate of  the hybrid coding scheme described in Section~\ref{sec:achieve},
\item The \emph{grouped} random coding strategy (described at the beginning of Section~\ref{sec:example}) wherein first, a sufficient number of random equations are sent such that Users~2 and 3 can decode $\msg_2$ and $\msg_3$, and then $\msg_1$ is sent,
\item The conventional random coding strategy (also described at the beginning of Section~\ref{sec:example}) wherein a sufficient number of random equations are sent such that all users can decode all messages,
\end{enumerate}

In Figures~\ref{fig:q1a} and~\ref{fig:q2b}, we demonstrate the gap between our BIC inner and outer bounds while focusing on varying the amount of information at the user with the \emph{least} side information. Figure~\ref{fig:q1a} demonstrates the gap between inner and outer bounds on symmetric capacity for a one-sided symmetric BIC problem. In particular, we fix $\mu_{21}=\mu_{23}=\frac{1}{2}$ and $\mu_{31}=\mu_{32}=\frac{1}{3}$ and consider the impact of varying $\mu_{12}=\mu_{13}=a$ across the range from $\frac{1}{2}$ to 1.
Figure~\ref{fig:q2b} demonstrates the gap between inner and outer bounds on symmetric capacity for a pairwise symmetric BIC problem. In particular, we fix $\mu_{23}=\mu_{32}=\frac{1}{3}$ and $\mu_{13}=\mu_{31}=\frac{1}{2}$
and consider the impact of varying $\mu_{12}=\mu_{21}=a^\prime$ across the range from $\frac{1}{2}$ to 1.

In Figures~\ref{fig:q1c} and~\ref{fig:q2d}, we demonstrate the gap between our BIC inner and outer bounds while focusing on varying the amount of information at the user with the \emph{most} side information. Specifically, in Figure~\ref{fig:q1c} we look at a one-sided symmetric scenario and fix $\mu_{12}=\mu_{13}=\frac{1}{2}$ and $\mu_{21}=\mu_{23}=\frac{1}{3}$, while varying $\mu_{31}=\mu_{32}=c$ across a range from 0 to $\frac{1}{2}$, while in Figure~\ref{fig:q2d} we look at the pairwise symmetric scenario and fix $\mu_{12}=\mu_{21}=\frac{2}{3}$ and $\mu_{13}=\mu_{31}=\frac{1}{2}$, while varying $\mu_{23}=\mu_{32}=c^\prime$ across a range from 0 to $\frac{1}{2}$. 

In the two BIC problems depicted in Figures~\ref{fig:q1a} and~\ref{fig:q2b}, we point out that as the user with the least amount of side information loses even more side information (increasing $a$ or $b$), the rate achievable by conventional random codes decreases. At some point in each Figures~\ref{fig:q1a} and~\ref{fig:q2b}, it is in fact to better to apply a grouped random coding strategy and assume that User~1 will not attempt to decode $\msg_2$ and $\msg_3$. On the other hand, in the BIC problems depicted in Figures~\ref{fig:q1c} and~\ref{fig:q2d}, since amount of side information of the least knowledgeable user remains constant (i.e., $\mu_{12}$ and $\mu_{13}$ are fixed), the rate achieved by conventional random coding is constant across the range. 

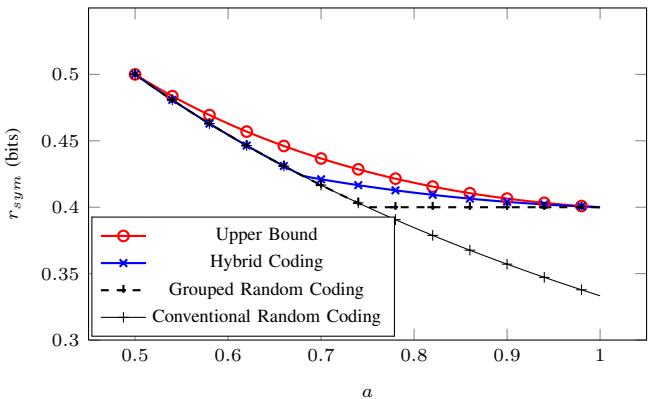
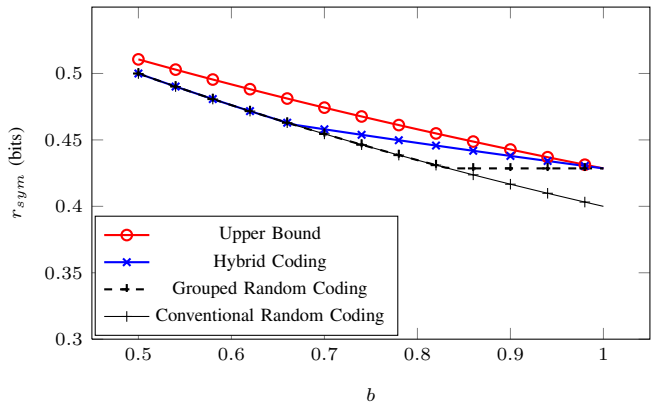
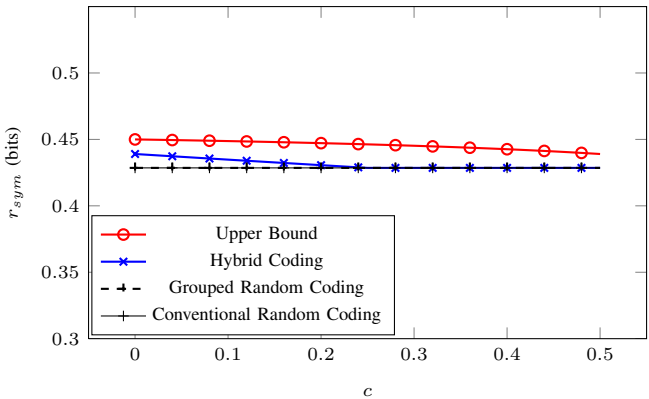
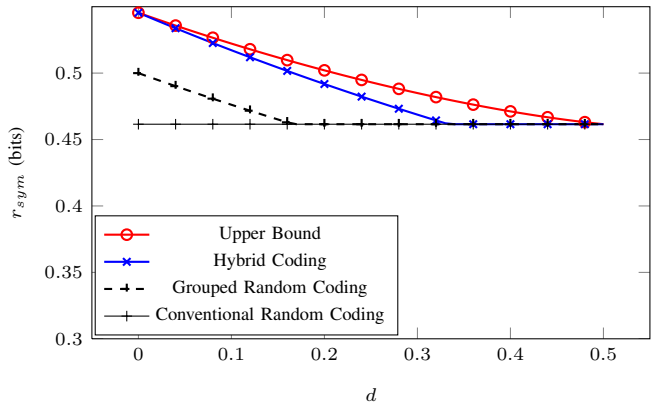
\begin{figure}[ht]
    \centering
\subfigure[$\mu_{12}=\mu_{13}=a$, $\mu_{21}=\mu_{23}=\frac{1}{2}$, and $\mu_{31}=\mu_{32}=\frac{1}{3}$ ]{
\begin{tikzpicture}[font=\scriptsize,%
    every axis/.style={
        ymax=0.55,%
        ymin=0.3,%
        ytick={0.3,0.35,...,0.55}}]
    \begin{axis}[
    width=9cm,
    height=6cm,
        xlabel=$a$,
        ylabel=$r_{sym}$ (bits),
        every axis y label/.style=
            {at={(ticklabel cs:0.5)},rotate=90,anchor=near ticklabel},
        legend style={at={(0.005,0.01)},anchor=south west}]
    \addplot[thick,color=red,mark=o,mark repeat=4]
        plot file {1a1213-UB.data};
    \addlegendentry{Upper Bound};
    \addplot[thick,color=blue,mark=x,mark repeat=4]
        plot file {1a1213-HC.data};
    \addlegendentry{Hybrid Coding};
    \addplot[color=black!50!black,mark=+,thick,dashed,mark repeat=4]
        plot file {1a1213-GRC.data};
    \addlegendentry{Grouped Random Coding};
    \addplot[color=black!50!black,mark=+,mark repeat=4]
        plot file {1a1213-RC.data};
    \addlegendentry{Conventional Random Coding};
    \end{axis}
\end{tikzpicture}\label{fig:q1a}
}
\subfigure[$\mu_{12}=\mu_{21}=b$, $\mu_{13}=\mu_{31}=\frac{1}{2}$, and $\mu_{23}=\mu_{32}=\frac{1}{3}$ ]{
\begin{tikzpicture}[font=\scriptsize,%
    every axis/.style={
        ymax=0.55,%
        ymin=0.3,%
        ytick={0.3,0.35,...,0.55}}]
    \begin{axis}[
    width=9cm,
    height=6cm,
        xlabel=$b$,
        ylabel=$r_{sym}$ (bits),
        every axis y label/.style=
            {at={(ticklabel cs:0.5)},rotate=90,anchor=near ticklabel},
        legend style={at={(0.005,0.01)},anchor=south west}]
    \addplot[thick,color=red,mark=o,mark repeat=4]
        plot file {2a1213-UB.data};
    \addlegendentry{Upper Bound};
    \addplot[thick,color=blue,mark=x,mark repeat=4]
        plot file {2a1213-HC.data};
    \addlegendentry{Hybrid Coding};
    \addplot[color=black!50!black,mark=+,thick,dashed,mark repeat=4]
        plot file {2a1213-GRC.data};
    \addlegendentry{Grouped Random Coding};
    \addplot[color=black!50!black,mark=+,mark repeat=4]
        plot file {2a1213-RC.data};
    \addlegendentry{Conventional Random Coding};
    \end{axis}
\end{tikzpicture}    \label{fig:q2b}
}\\
\subfigure[$\mu_{12}=\mu_{13}=\frac{2}{3}$, $\mu_{21}=\mu_{23}=\frac{1}{2}$, and $\mu_{31}=\mu_{32}=c$ ]{
\begin{tikzpicture}[font=\scriptsize,%
    every axis/.style={
        ymax=0.55,%
        ymin=0.3,%
        ytick={0.3,0.35,...,0.55}}]
    \begin{axis}[
    width=9cm,
    height=6cm,
        xlabel=$c$,
        ylabel=$r_{sym}$ (bits),
        every axis y label/.style=
            {at={(ticklabel cs:0.5)},rotate=90,anchor=near ticklabel},
        legend style={at={(0.005,0.01)},anchor=south west}]
    \addplot[thick,color=red,mark=o,mark repeat=4]
        plot file {1c1213-UB.data};
    \addlegendentry{Upper Bound};
    \addplot[thick,color=blue,mark=x,mark repeat=4]
        plot file {1c1213-HC.data};
    \addlegendentry{Hybrid Coding};
    \addplot[color=black!50!black,mark=+,thick,dashed,mark repeat=4]
        plot file {1c1213-GRC.data};
    \addlegendentry{Grouped Random Coding};
    \addplot[color=black!50!black,mark=+,mark repeat=4]
        plot file {1c1213-RC.data};
    \addlegendentry{Conventional Random Coding};
    \end{axis}
\end{tikzpicture}
    \label{fig:q1c}
}
\subfigure[$\mu_{12}=\mu_{21}=\frac{2}{3}$, $\mu_{13}=\mu_{31}=\frac{1}{2}$, and $\mu_{23}=\mu_{32}=d$ ]{
\begin{tikzpicture}[font=\scriptsize,%
    every axis/.style={
        ymax=0.55,%
        ymin=0.3,%
        ytick={0.3,0.35,...,0.55}}]
    \begin{axis}[
    width=9cm,
    height=6cm,
        xlabel=$d$,
        ylabel=$r_{sym}$ (bits),
        every axis y label/.style=
            {at={(ticklabel cs:0.5)},rotate=90,anchor=near ticklabel},
        legend style={at={(0.005,0.01)},anchor=south west}]
    \addplot[thick,color=red,mark=o,mark repeat=4]
        plot file {2c1213-UB.data};
    \addlegendentry{Upper Bound};
    \addplot[thick,color=blue,mark=x,mark repeat=4]
        plot file {2c1213-HC.data};
    \addlegendentry{Hybrid Coding};
    \addplot[color=black!50!black,mark=+,thick,dashed,mark repeat=4]
        plot file {2c1213-GRC.data};
    \addlegendentry{Grouped Random Coding};
    \addplot[color=black!50!black,mark=+,mark repeat=4]
        plot file {2c1213-RC.data};
    \addlegendentry{Conventional Random Coding};
    \end{axis}
\end{tikzpicture}
    \label{fig:q2d}
}
\caption{Inner and outer bounds on the symmetric capacity of example 3-user BIC problems: (a) One-sided side information symmetry, and (b) pairwise side information symmetry, while varying the least knowledgeable user's side information under; and
(c) one-sided side information symmetry, and (d) pairwise side information symmetry, while varying the most knowledgeable user's side information. }\label{fig:q}
\end{figure}

With the figures, we highlight the following observations about our inner and outer bounds:
\begin{enumerate}
    \item There exists a threshold for side information parameters where below this threshold, in the best hybrid coding strategy all three users decode all messages and thus the achieved rate is the same as conventional random codes. In particular, this is true for small $a$ and $b$ in Figures~\ref{fig:q1a} and \ref{fig:q2b} and larger $c$ and $d$ in Figures~\ref{fig:q1c} and \ref{fig:q2d}, respectively. However, beyond this threshold (larger $a$ and $b$ and smaller $c$ and $d$), we observe a clear potential for increased rate from hybrid codes. It is worth noting that the regimes where hybrid codes offer a rate increase are those further from the fully symmetric BIC problem (where all network parameters, $\mu_{ij}$, are the same). Recall that for the fully symmetric BIC problem the entire capacity region is achievable using conventional random coding (see Proposition~\ref{prop:KuSym}). 
    \item In Figures~\ref{fig:q1a} and~\ref{fig:q2b}, when $a=1$ or $b=1$ there exist no opportunities at all to exploit the side information at User~1. Hence, both hybrid coding and grouped random coding achieve the genie upper bound.
    \item Although there exists a gap between our inner and outer bounds, we highlight a specific case where our new hybrid coding scheme both provides strictly positive rate gain over conventional random coding \emph{and meets the new upper bound}: in Figure~\ref{fig:q2d} when $d=0$. This scenario is related to the one considered in the motivating example of Section~\ref{sec:example}, in the sense that Users~2 and 3 know each other's complete message as side information. 
\end{enumerate}

\section{Blind Index Coding over Wireless Channels}\label{sec:BICW}

In this section, we generalize the BIC problem model further to consider the impact of uncertainty not only within the side information given to users, but also in the sender-to-user broadcast channel (recall that in the BIC problem this channel was error free). In particular, we emulate loss of packetized transmissions due to fading in wireless channels using a binary fading model for the sender-to-user broadcast. Consequently, the problem considered here will be referred to as blind index coding over wireless channels (BICW). 

As we will see, considering wireless transmissions adds new challenges to the problem, and surprisingly repetition of uncoded bits (within the hybrid coding framework) will become a powerful technique for increasing achievable rate. Unlike the BIC problem considered in the previous sections, even the 2-user BICW problem is nontrivial. Hence, in this section we focus on a a 2-user problem representative of general BICW problems. After formally defining the representative problem, we define a hybrid coding scheme that not only XORs randomc combinations of some messages with uncoded bits of others, but also uses repetition of uncoded bits. We derive the achievable rate regions of these hybrid codes with repetitions, and then denomstrate numerically the resulting gain in achievable rate that our scheme provides over conventional methods.     

\subsection{Wireless Broadcast Channel Model}
\label{sec:BICWmod}
In the BICW scenario the channel output received by by User~$i$, $\vec{y}_i$, is governed by a binary fading process. Specifically, let $\vec{\gamma}_{i}$ be a binary vector with the same length as the channel input vector $\vec{x}$ and drawn i.i.d from a Bernoulli$(1-\epsilon_{i})$ distribution. The channel output for User~$i$ is given by the input-output relationship
\begin{align}
    y_{i}[\ell] = \gamma_{i}[\ell]{x}[\ell].
\end{align}
User~$i$ knows $\vec{\gamma}_{i}$, however the sender is only aware of parameters $\{\epsilon_{i}\}$, which govern the probabilistic behavior of the sender-to-user broadcast channel.

In this section, we assume the model depicted in Figure~\ref{fig:BICW}, containing only two users where $\epsilon_1 < \epsilon_2$, $\mu_{12}=1$, and $\mu_{21}=\mu$ (i.e., User~1 has a better channel than User~2 but no side information). 
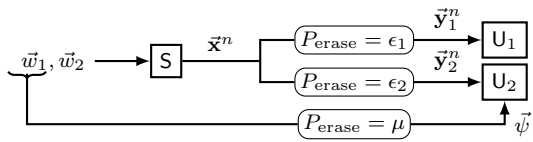
\begin{figure}[ht]
\centering
\begin{tikzpicture}[yscale=0.28,font=\footnotesize]
\node (m) at (0,0) [] {$\msg_1, \msg_2$};
\node (s) at (1.5,0) [draw,thick] {$\sender$};
\node (e1) at (4,1) [draw,rounded corners,inner sep=2pt] {$P_\mathrm{erase} = \epsilon_1$};
\node (e2) at (4,-1) [draw,rounded corners,inner sep=2pt] {$P_\mathrm{erase} = \epsilon_2$};
\node (eSI) at (4,-3) [draw,rounded corners,inner sep=2pt] {$P_\mathrm{erase} = \mu$};
\node (d1) at (6,1) [draw,thick] {$\user_1$};
\node (d2) at (6,-1) [draw,thick] {$\user_2$};
\draw[thick,-latex] (m) -- (s);
\draw[thick] (s) -- node[above] {$\vec{\mathbf{x}}^n$} (2.75,0);
\draw[thick] (2.75,0) |- (e1);
\draw[thick] (2.75,0) |- (e2);
\draw[thick,-latex] (e1) -- node[above] {$\vec{\mathbf{y}}_1^n$}(d1);
\draw[thick,-latex] (e2) -- node[above] {$\vec{\mathbf{y}}_2^n$}(d2);

\draw[decorate,decoration={brace},thick] (-0.1,-0.5) -- (-0.6,-0.5);

\draw[thick] (-0.35,-0.6) |- (eSI);
\draw[thick,-latex] (eSI) -| node[right] {$\vec{\mathbf{\psi}}$}(d2);
\end{tikzpicture}
\caption{2-user instance of the BICW problem.}\label{fig:model2}
\end{figure}

\begin{remark}
We assume that $\epsilon_1 < \epsilon_2$ and that side information was only given to User~2 (i.e., $\mu_{21}=1$) for ease of exposition. In all other 2-user settings (i.e., arbitrary $\epsilon_1$ and $\epsilon_2$ and side information at either user), either there is no index coding gain even if the sender knows the side information or the natural generalization of our proposed scheme recovers some index coding gain to outperform conventional approaches.
\end{remark}

Our main result for this setting is as follows.
\begin{theorem}\label{thm:HRC}
For the 2-user BICW problem defined above, the rate region $\mathcal{R}$ is achievable, where $\mathcal{R}$ is the set of all non-negative rate pairs $(r_1,r_2)$ satisfying,
\begin{align}
    r_1+r_2 \leq{}& 1-\epsilon_1,\label{eq:ratereg_A}\\
    \omega_1(L)r_1 +\omega_2(L) r_2 \leq{}& 1-\epsilon_2,\quad L=1,\ldots,L_{max}\label{eq:ratereg_B}
\end{align}
where
\begin{align}
\omega_1(L) ={}&  \frac{1-\epsilon_2}{1-\epsilon_1}\epsilon_1^{L} + \mu(1-\epsilon_2^{L})\omega_2(L) 
    + L(1-\epsilon_2)\left(1 -\omega_2(L)\right),\label{eq:weight1}\\
\omega_2(L) ={}& \min\left\{\frac{1-\epsilon_1^{L}}{1-\mu\epsilon_2^{L}},1\right\},\label{eq:weight2}\\
L_{max} \triangleq{}& 1+\left\lfloor\frac{\log(\mu)}{\log(\epsilon_1/\epsilon_2)}\right\rfloor.\label{eq:Lmax}
\end{align}
\end{theorem}

\begin{remark}
Notice that as $\epsilon_2\rightarrow 0$ (and by the assumption $\epsilon_2>\epsilon_1$, as $\epsilon_1\rightarrow 0$), the BICW problem reverts to a BIC problem. Moreover as $\epsilon_2\rightarrow 0$, $\omega_1(L)\rightarrow\mu$
and $\omega_2(L)\rightarrow 1$, resulting in the achievable region of rate pairs satisfying:
\begin{align*}
    r_1+r_2\leq{}& 1,\\
    \mu r_1+r_2\leq{}& 1,
\end{align*}
which is equivalent (given assumptions on $\mu_{12}$ and $\mu_{21}$) to the 2-user BIC capacity region (formally stated in Proposition~\ref{prop:2u}). 
\end{remark}

\subsection{Proof of Theorem~\ref{thm:HRC}}
\label{sec:ach}
This section is organized as follows. We first define the hybrid coding scheme by specifying a class of generator matrices which map length-$m$ message vectors to length-$n$ codewords, and which are parametrized by three quantities: $\rho$, $L$, and $\alpha$. For each $n$, the transmitter maps two messages, $\msg_1$ and $\msg_2$ to codewords using corresponding generator matrices (with different parameters), and XORs the two codewords to produce the channel input vector.
We then specify the method of decoding and establish the achievable rate region for our coding scheme when fixing the generator matrix parameters for all $n$. By doing so, we show that for any $(r_1,r_2)\in \mathcal{R}$ (as defined in Theorem~\ref{thm:HRC}) there exists a choice of parameters such that $(r_1,r_2)$ is achievable, thus proving Theorem~\ref{thm:HRC}.

\subsubsection{Encoding}

Our hybrid coding scheme encodes $\msg_1$ and $\msg_2$ separately and linearly, before combining the resulting codewords through bit-wise XOR. 
The codeword for each message is constructed in a manner similar to the component of BIC hybrid codes from the previous section specific to a single message component: uncoded repetitions of message bits are supplemented by a random linear combinations. The specific mapping from message to codeword is formalized in the following definition, parametrized for a given $n$ by three quantities $\rho$, $L$ and $\alpha$:
\begin{definition}[Repetition plus Random Parity (RRP) Matrix]
\label{def:RRP}
An $n\times m$ RRP matrix with parameters $\rho\in [0,1]$, $L\in\mathbb{N}$, and $\alpha\in[0,1]$ is a binary matrix, $\mathbf{U}$, with the form:
\begin{align}
\mathbf{U} ={}& \begin{bmatrix}
\mathbf{B}^\top &
\mathbf{A}_{1}^\top &
\ldots &
\mathbf{A}_{L+1}^\top &
\mathbf{0} 
\end{bmatrix}^\top,\label{eq:U1}
\end{align}
where\vspace{-0.25cm}
\begin{align*}
\mathbf{A}_\ell = \begin{cases}
    \mathbf{I}_{m} & \text{ if } \ell\leq L\cr
    [\mathbf{I}_{\alpha m}\quad \mathbf{0}] & \text{ else}
    \end{cases},
\end{align*}
and $\mathbf{B}$ is a $\rho n\times m$ matrix with entries drawn i.i.d. from $\bern$. For feasibility, we require that $\alpha m$ is an integer, and
\begin{align}
    (L+\alpha)\frac{m}{n}+\rho \leq 1. \label{eq:feasible}
\end{align}
\end{definition}

\begin{remark}
Simply stated, an RRP matrix maps a length-$m$ message vector to a length-$n$ codeword by repeating each uncoded message bit either $L$ or $L+1$ times. The parameter $\alpha$ specifies the fraction of bits repeated $L+1$ times, while $\rho$ specifies the proportion of length-$n$ codeword reserved for random linear coded parity. Inequality (\ref{eq:feasible}) ensures that $\mathbf{U}$ is an $n\times m$ matrix.

It is worth noting that in the hybrid encoding scheme described for 3-user (non-wireless) BIC, the mapping of message $\msg_1$, $\msg_2$, and $\msg_3$ to sequences before XOR (i.e., the individually colored bars in Figure~\ref{fig:3coding}) could be interpreted as RRP matrices. For $\msg_1$, we chose $L=\alpha=0$ and for messages $\msg_2$ and $\msg_3$ we chose $L=1$ and $\alpha=0$.
The use of RRP matrices with $L>1$ and $\alpha>0$ (i.e., the \emph{repetition} of uncoded message bits) is the key innovation to hybrid coding that enables higher rate in the wireless setting.
\end{remark}

Using the defined RRP matrices, we now describe the encoding scheme that maps messages $\msg_1$ and $\msg_2$ to a length-$n$ channel input vector.
Let $n$, $m_1^{(n)}$, and $m_2^{(n)}$ be given. For each $n$, let $\mathbf{U}_1$ be a $n\times m_1^{(n)}$ RRP matrix with parameters $(\rho_1,L_1,\alpha_1)$ and $\mathbf{U}_2$ be a $n\times m_2^{(n)}$ RRP matrix with parameters $(\rho_2,L_2,\alpha_2)$. 
The channel input vector $\vec{\mathbf{x}}^n$ is given by (assuming modulo-2 addition):
\begin{align*}
    \vec{\mathbf{x}}^n  
        ={}& \begin{bmatrix}\mathbf{U}_1 & \mathbf{U}_2\end{bmatrix}\begin{bmatrix}\msg_1 \\ \msg_2\end{bmatrix}\\
        ={}& \mathbf{U}_1\msg_1 + \mathbf{U}_2\msg_2.
\end{align*}

Figure~\ref{fig:BICWcoding} depicts an example hybrid encoding with repetitions for the 2-user BICW setting. In this particular example, $L=2$ and $\alpha=0.5$.
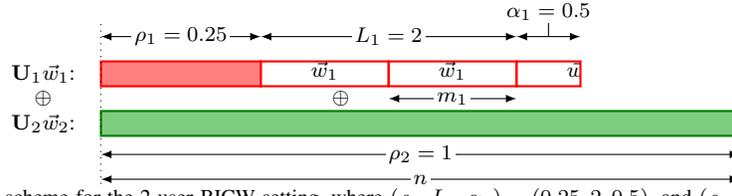
\begin{figure}
\centering
\begin{tikzpicture}[xscale=1.7,yscale=0.66,font=\footnotesize]

\filldraw [draw=green!50!black,fill=green!60!black!50,thick] (0,1) rectangle (5,1.5);
\draw (-0.45,1.25) node[] {$\mathbf{U}_2\msg_2$:\ };
\draw (-0.45,1.75) node[] {$\oplus$\ };

\draw (3.75,2.312) node[]{\footnotesize $\msg_1$};
\fill[white] (3.75,2) rectangle (4.75,2.5);
\filldraw [draw=red,fill=red!50,thick] (0,2) rectangle (1.25,2.5);
\draw [draw=red,thick] (1.25,2) rectangle (2.25,2.5);
\draw [draw=red,thick] (2.25,2) rectangle (3.25,2.5);
\draw [draw=red,thick] (3.25,2) rectangle (3.75,2.5);
\draw (1.75,2.312) node[]{\footnotesize $\msg_1$};
\draw (2.75,2.312) node[]{\footnotesize $\msg_1$};
\draw (-0.45,2.25) node[] {$\mathbf{U}_1\msg_1$:\ };

\draw[latex-latex,thin] (0,3) --node [fill=white,inner sep=1pt]{$\rho_1=0.25$} (1.25,3);
\draw[latex-latex,thin] (1.25,3) --node [fill=white,inner sep=1pt]{$L_1=2$} (3.25,3);
\draw[latex-latex,thin] (3.25,3) -- (3.75,3);
\draw[thin] (3.5,3) -- (3.5,3.5) node [fill=white,inner sep=1pt]{$\alpha_1=0.5$};
\draw[latex-latex,thin] (2.25,1.75) --node [fill=white,inner sep=1pt]{$m_1$} (3.25,1.75);
\draw[latex-latex,thin] (0,0.625) --node [fill=white,inner sep=1pt]{$\rho_2=1$} (5,0.625);
\draw[latex-latex,thin] (0,0.125) --node [fill=white,inner sep=1pt]{$n$} (5,0.125);
\draw[dotted,thin] (0,3.25) -- (0,0);
\draw (1.875,1.75) node {$\oplus$};
\draw[dotted,thin] (5,3.25) -- (5,0);
\end{tikzpicture}\vspace{-2.5ex}
\caption{An example hybrid coding scheme for the 2-user BICW setting, where $(\rho_1,L_1,\alpha_1)=(0.25,2,0.5)$, and $(\rho_2,L_2,\alpha_2)=(1,0,0)$. Outlined boxes represent uncoded bits, shaded boxes represent RLCs of a single message.
\vspace{-4.5ex}}\label{fig:BICWcoding}
\end{figure}

\subsubsection{Decoding}

We now specify the decoding strategy and then characterize the achievable rates for our scheme with fixed parameters $\rho_i$, $L_i$ and $\alpha_i$, $i=1,2$. In what follows, we choose $(\rho_2,L_2,\alpha_2) =(1,0,0)$ (i.e., User~2's generator matrix, $\mathbf{U}_2$, is a random matrix). Choosing parameters $(\rho_1,L_1,\alpha_1)$ is more nuanced and will be addressed within the analysis. For brevity, we will not explicitly analyze the error rates of our scheme for given $n$, but instead provide a sketch of the achievability proof using existing results for random linear codes over point-to-point erasure channels.

In our decoding strategy, User~1 first decodes $\msg_2$ and peels its interfering contribution from its received signal, and then decodes its desired message, $\msg_1$. User~2 only decodes $\msg_2$. We first describe decoding $\msg_2$ at each user.

Recall that the channel input at any time, $t$, is given by 
\mbox{$\mathbf{x}[t]=\mathbf{U}_1(t,:)\msg_1 + \mathbf{U}_2(t,:)\msg_2$},
where $\mathbf{U}_i(t,:)$ is the $t$-th row of generator matrix $\mathbf{U}_i$. 
The decoding strategy for $\msg_2$ used by both users is based on the following observation. If $t$ and $t^\prime\neq t$ both correspond to a repetition of the same message bit from $\msg_1$, then the modulo-2 sum of these yields
\mbox{$\mathbf{x}[t]+\mathbf{x}[t^\prime] = (\mathbf{U}_2(t,:)+\mathbf{U}_2(t^\prime,:))\msg_2$},
which is a random linear combination of only $\msg_2$ bits (since $\rho_2=1$). By this method we ``clean''  equations of $\msg_1$. User~2 has the additional option of using its side information to clean equations, which has the same essence. 

The cleaned random linear equations are used by each user in conjunction with those that by construction were only functions of $\msg_2$ (i.e., for those $t$ where in (\ref{eq:U1}) $\mathbf{U}_1(t,:)=0$) to decode $\msg_2$.
After decoding $\msg_2$, User~1 removes the contribution of $\msg_2$ from its received signal before decoding $\msg_1$. 
If any of these decodings fail, then an error occurs. We now claim that the decoding scheme yields the following achievable rates, proven in Appendix~\ref{app:lem:achieveBICW}:
\begin{lemma}\label{lem:achieveBICW}
Consider the 2-user BWIC problem defined by network parameters $\epsilon_1$, $\epsilon_2$, and $\mu$, and let $\rho_1\in[0,1]$, $L_1\in\mathbb{N}$, and $\alpha_1\in[0,1)$ be fixed. A rate pair $(r_1,r_2)$ is achievable if it satisfies,%
\begin{align}
r_1 \leq{}& \frac{1-\rho_1}{L_1+\alpha_1},\label{eq:achJLA0}\\
r_1 \leq{}& \rho_1\frac{1-\epsilon_1}{\epsilon_1^{L_1}-\alpha_1(\epsilon_1^{L_1}-\epsilon_1^{L_1+1})},\label{eq:achJLA1}\\
[1-\epsilon_1^{L_1}+\alpha_1(\epsilon_1^{L_1}-\epsilon_1^{L_1+1})]r_1+r_2 {}\leq{}& (1-\epsilon_1)(1-\rho_1),\label{eq:achJLA2}\\
\mu[1-\epsilon_2^{L_1}+\alpha_1(\epsilon_2^{L_1}-\epsilon_2^{L_1+1})]r_1+r_2 {}\leq{}& (1-\epsilon_2)(1-\rho_1).\label{eq:achJLA3}
\end{align}
\end{lemma}

From Lemma~\ref{lem:achieveBICW}, it is clear that by considering the union or achievable rate pairs over all $(\rho_1,L_1,\alpha_1)$ we arrive at the rate region achievable by our schemes. Specifically, let $\mathcal{R}(\rho_1,L_1,\alpha_1)$ for $\rho_1\in[0,1]$, $L_1\in\mathbb{N}$, and $\alpha_1\in[0,1]$ be defined as the set of all pairs $(r_1,r_2)$ satisfying (\ref{eq:achJLA0})--(\ref{eq:achJLA3}), and we define a rate region: 
\begin{align}
\overline{\mathcal{R}}\triangleq \bigcup_{\rho_1,L_1,\alpha_1} \mathcal{R}(\rho_1,L_1,\alpha_1).\label{eq:BICW_ratereg_bar}
\end{align}

To complete the proof of Theorem~\ref{thm:HRC}, we now demonstrate that the region $\mathcal{R}$ (as defined in Theorem~\ref{thm:HRC}) is contained within $\overline{\mathcal{R}}$ (given in (\ref{eq:BICW_ratereg_bar})), and thus is achievable. To do so, we need only show that for every rate pair $(r_1,r_2)\in\mathcal{R}$, there exists parameters $(\rho_1,L_1,\alpha_1)$ such that (\ref{eq:achJLA0})--(\ref{eq:achJLA3}) are satisfied. We therefore fix $r_1$ to any value in the interval $[0,1-\epsilon_1]$, and choose parameters $\rho_1^*$, $L_1^*$, and $\alpha_1^*$ as
\begin{align}
L_1^* ={}& 
        \begin{array}[t]{cl}
        \maximize & \min\{L,L_{max}\}\\
        \subjectto & L\in\mathbb{N} \\
         & \frac{\epsilon_1^{L}}{1-\epsilon_1} r_1 \leq 1 - L r_1\end{array},\label{eq:opar2}\\
\alpha_1^* ={}& \begin{cases}
    0 & \text{ if }L_1^*=L_{max}\cr
    \frac{1-r_1\left(\frac{\epsilon_1^{L_1^*}}{1-\epsilon_1}+L_1^*\right)}{r_1(1-\epsilon_1^{L_1^*})} & \text{ if }L_1^*<L_{max}
\end{cases},\label{eq:opar3}\\
\rho_1^* ={}& \frac{\epsilon_1^{L_1^*}-\alpha_1^*(\epsilon_1^{L_1^*}-\epsilon_1^{L_1^*+1})}{1-\epsilon_1}r_1,\label{eq:opar1}
\end{align}
where $L_{max}$ is as defined in (\ref{eq:Lmax}).
Notice that given $r_1$, we first determine the appropriate $L_1^*$, then $\alpha_1^*$, then finally $\rho_1^*$ and that both (\ref{eq:achJLA0}) and (\ref{eq:achJLA1}) are satisfied by the chosen parameters.

Substituting these into (\ref{eq:achJLA2}) and (\ref{eq:achJLA3}), we see that $r_2$ is achievable if it satisfies both of the following inequalities:
\begin{align}
    r_2 
        \leq{}& (1-\epsilon_1)-\rho_1^*(1-\epsilon_1) 
            - r_1\left[1-\epsilon_1^{L_1^*}+\alpha_1^* (\epsilon_1^{L_1^*}-\epsilon_1^{L_1^*+1}) \right]\nonumber\\
        ={}& 1-\epsilon_1-r_1,\label{eq:thmproof0a}\\
    r_2 \leq{}& (1-\epsilon_2)-\rho_1^*(1-\epsilon_2) - r_1\mu\left[1-\epsilon_2^{L_1^*}+\alpha_1^* (\epsilon_2^{L_1^*}-\epsilon_2^{L_1^*+1}) \right]
            \nonumber\\
        ={}& (1-\epsilon_2)\left(1 - r_1\left[\frac{\epsilon_1^{L_1^*}}{1-\epsilon_1}+\mu\frac{1-\epsilon_2^{L_1^*}}{1-\epsilon_2} - \alpha_1^*\left(\epsilon_1^{L_1^*}-\mu\epsilon_2^{L_1^*}\right)\right]\right)
            \nonumber\\
        \stackrel{(a)}{=}{}& \frac{1-\epsilon_2-\omega_1(L_1^*)r_1}{\omega_2(L_1^*)},\label{eq:thmproof0b}
\end{align}
where in (a) we compared the evaluated expression with $\omega_1(L)$ and $\omega_2(L)$ as defined in (\ref{eq:weight1}) and (\ref{eq:weight2}) evaluated at $L_1=L_1^*$. We now point out that (\ref{eq:thmproof0a}) is equivalent to (\ref{eq:ratereg_A}) and (\ref{eq:thmproof0b}) is equivalent to (\ref{eq:ratereg_B}) evaluated at $L=L_1^*$. Moreover, since
\begin{align}
\frac{1-\epsilon_2-\omega_1(L_1^*)r_1}{\omega_2(L_1^*)}
    \geq{}& \min_L\ \frac{1-\epsilon_2-\omega_1(L)r_1}{\omega_2(L)},\label{eq:thmproof0}
\end{align}
and the right hand side of inequality represents the tightest version of (\ref{eq:ratereg_B}) for fixed $r_1$, we observe that any $(r_1,r_2)$ satisfying (\ref{eq:ratereg_A}) and (\ref{eq:ratereg_B}) for all $L\leq L_{max}$ (i.e., any $(r_1,r_2) \in \mathcal{R}$) is indeed achievable, thus completing the proof of Theorem~\ref{thm:HRC}.

\subsection{Numerical Results}
For blind index coding over wireless channels, we recall that the key difference was the usefulness of \emph{repeating} uncoded bits within the hybrid coding scheme. Therefore, we now provide numerical results for three BICW scenarios, characterized by $\epsilon_1$, $\epsilon_2$, and $\mu$. In each, we plot $\mathcal{R}$ and highlight regimes (along x-axes) wherein the number of repetitions used in our scheme increases. For each scenario, we point out the gain in $r_2$ offered by repetion-based hybrid codes over conventional schemes, and for further comparison we also depict rate regions achieved by: 1) Conventional random codes as defined in the beginning of Section~\ref{sec:BIC:achieve}, 2) Time-Division between separate random encoding of $\msg_1$ and $\msg_2$, and 3) the following genie-aided upper bound:
\begin{proposition}
For the 2-user BICW problem setting considered in Theorem~\ref{thm:HRC}, an achievable rate pair $(r_1,r_2)$ must satisfy
\begin{align}
    \max\left\{r_1 + r_2,\mu r_1 + \frac{1-\epsilon_1}{1-\epsilon_2} r_2\right\} \leq{}& 1-\epsilon_1.
\end{align}
\end{proposition}
\begin{proof}
The bound may be separated into two outer bounds that correspond to the first and second terms within the $\max$, respectively:
\begin{itemize}
    \item $r_1+r_2 \leq 1-\epsilon_1$,
    \item $\mu \frac{1-\epsilon_2}{1-\epsilon_1} r_1 + r_2 \leq 1-\epsilon_2$.
\end{itemize}
Denote the subvector of $\msg_1$ given as side information as $\msg_1^+$ and the complementary subvector as $\msg_1^-$.  We prove the first bound by applying Fano's inequality at each user to observe:
\begin{align}
nr_1 \leq{}& I\left(\vec{y}_1,\vec{\gamma}_1;\msg_1\right) + o(n)\nonumber\\
    ={}& I\left(\vec{y}_1,\vec{\gamma}_1;\msg_1\right) + o(n)\nonumber\\
    ={}& H\left(\vec{y}_1|\vec{\gamma}_1\right) - H\left(\vec{y}_1|\vec{\gamma}_1,\msg_1\right)  + o(n)\nonumber\\
    \stackrel{(a)}{\leq}{}& H\left(\vec{y}_1|\vec{\gamma}_1\right) - H\left(\vec{y}_2|\vec{\gamma}_2,\msg_1\right)  + o(n)\nonumber\\
    \leq{}& H\left(\vec{y}_1|\vec{\gamma}_1\right) - H\left(\vec{y}_2|\vec{\gamma}_2,\msg_1^+\right)  + o(n)\nonumber\\
    \leq{}& n(1-\epsilon_1) - H\left(\vec{y}_2|\vec{\gamma}_2,\msg_1^+\right)  + o(n), \label{eq:genieprooffano1}\\
nr_2 \leq{}& I\left(\vec{y}_2,\vec{\gamma}_2,\sinfo_{21},\vec{g}_{21};\msg_2\right) + o(n)\nonumber\\
    ={}&H\left(\vec{y}_2|\vec{\gamma}_2,\sinfo_{21},\vec{g}_{21}\right)- H\left(\vec{y}_2|\vec{\gamma}_2,\sinfo_{21},\vec{g}_{21},\msg_2\right) + o(n)\nonumber\\
    \leq{}&H\left(\vec{y}_2|\vec{\gamma}_2,\sinfo_{21},\vec{g}_{21}\right) + o(n)\nonumber\\
    ={}&H\left(\vec{y}_2|\vec{\gamma}_2,\msg_1^+\right) + o(n), \label{eq:genieprooffano2}
\end{align}
where in step (a) we observed that because the sender does not know the fading channel state of the sender-to-user channel. We complete the proof of the first bound by combining (\ref{eq:genieprooffano1}) and (\ref{eq:genieprooffano2}) and normalizing by $n$ as $n$ grows large.

To prove the second bound, we consider a genie which provides the sender of knowledge regarding which bits of $\msg_1$ are given as side information to User~2. We again applying Fano's inequality at each user, but in a different way, to observe
\begin{align}
n\mu r_1 
    \leq{}& I\left(\vec{y}_1,\vec{\gamma}_1;\msg_1^- \right) + o(n)\nonumber\\
    \leq{}& I\left(\vec{y}_1,\vec{\gamma}_1,\msg_1^+,\msg_2;\msg_1^- \right) + o(n)\nonumber\\
    \leq{}& H\left(\vec{y}_1\middle| \vec{\gamma}_1,\msg_1^+,\msg_2\right) + o(n),\label{eq:genieprooffano3}\\
nr_2 \leq{}& I\left(\vec{y}_2,\vec{\gamma}_2,\sinfo_{21},\vec{g}_{21};\msg_2\right) + o(n)\nonumber\\
    ={}&H\left(\vec{y}_2|\vec{\gamma}_2,\sinfo_{21},\vec{g}_{21}\right)- H\left(\vec{y}_2|\vec{\gamma}_2,\sinfo_{21},\vec{g}_{21},\msg_2\right) + o(n)\nonumber\\
    \leq{}& n(1-\epsilon_2) - H\left(\vec{y}_2|\vec{\gamma}_2,\msg_1^+,\msg_2\right) + o(n),\nonumber\\
    \stackrel{(b)}{\leq}{}& n(1-\epsilon_2) - \frac{1-\epsilon_1}{1-\epsilon_2}H\left(\vec{y}_1|\vec{\gamma}_1,\msg_1^+,\msg_2\right) + o(n),\label{eq:genieprooffano4}
\end{align}
where in step (b) we applied Lemma~1 of~\cite{VMA2014} which when applied to our problem states that (because the sender does not know the binary fading channel states $\{\vec{\gamma}_i\}$),
\begin{align*}
    H\left(\vec{y}_2|\vec{\gamma}_2,\msg_1^+,\msg_2\right)\geq{}&\frac{1-\epsilon_1}{1-\epsilon_2}H\left(\vec{y}_1|\vec{\gamma}_1,\msg_1^+,\msg_2\right).
\end{align*}
To complete the proof of the second outer bound, we scale (\ref{eq:genieprooffano3}) by $\frac{1-\epsilon_2}{1-\epsilon_1}$ and combine with (\ref{eq:genieprooffano4}).
\end{proof}

In Figure~\ref{fig:p1}, $\epsilon_1=\frac{1}{2}$, $\epsilon_2=\frac{3}{4}$, $\mu=\frac{1}{2}$ notice that when $r_1$ is near the point-to-point capacity of 0.5, hybrid coding recovers all of the available index coding gain. This is because when $r_1$ is near 0.5, the primary challenge is not blindly exploiting side information, but rather accounting for interference incurred at User~1. For this set of network parameters, we point out that for any fixed value of $r_1$, hybrid coding offers at least 62\% of the available index coding gain.

In Figure~\ref{fig:p2}, $\epsilon_1=\frac{1}{2}$, $\epsilon_2=\frac{9}{10}$, $\mu=\frac{1}{10}$ we consider a BICW setting where side information is plentiful (User~2 knows 90\% of $\msg_1$). In this case, $L_{max}=4$ and the piece-wise linear boundary of the hybrid coding achievable rate region has more linear segments, with segments corresponding to the number of repetitions used. For this setting and for any fixed $r_1$, HRC always achieves at least 68\% of the available index coding gain.

Finally, in Figure~\ref{fig:p3}, $\epsilon_1=\frac{1}{2}$, $\epsilon_2=\frac{3}{4}$, $\mu=\frac{9}{10}$ we consider a BICW setting with very little side information (User~2 knows 10\% of $\msg_1$). In this case, $L_{max}=1$ and from the figure, it is apparent that although any index coding gain is modest, it is still strictly positive for all $r_1\notin\{0,1-\epsilon_1\}$.

\begin{figure}[ht]
    \centering
\subfigure[ $\epsilon_1=\frac{1}{2}$, $\epsilon_2=\frac{3}{4}$, $\mu=\frac{1}{2}$]{
\begin{tikzpicture}[xscale=13,yscale=13,font=\scriptsize]
\filldraw[draw=black,fill=white] (0,0.25) -- node [above,pos=0.7,sloped]{Genie-Aided Upper Bound} (0.3333,0.1667) -- (0.5,0) -- (0,0) -- cycle;
\filldraw[draw=red,fill=red!4!white] (0,0.25) --  (0.39,0.11)node [anchor=north east,red,inner sep=1pt]{Hybrid} -- (0.5,0)  -- (0,0) -- cycle;
\filldraw[draw=green!70!black,dashed,fill=green!70!black!8!white] (0,0.25) -- (0.5,0) -- (0,0) -- cycle;
\node at(0.255,0.075)[above,green!70!black,inner sep=2pt] {Random Code};
\node at(0.255,0.075)[green!70!black,inner sep=2pt] {+};
\node at(0.255,0.075)[below,green!70!black,inner sep=2pt] {Time Division};
\filldraw[draw=blue,dotted,fill=blue!10!white] (0,0.25) -- (0.25,0) -- (0.5,0) -- (0,0) -- cycle;
\node at(0.1,0.02)[above,blue,inner sep=2pt] {Conventional Random Code};
\draw[] (0,0) rectangle (0.55,0.3);
\node at (0.275,-0.05)[below] {$r_1$};
\node at (-0.05,0.15)[rotate=90,above] {$r_2$};
\draw[] (0.1,0.01) -- (0.1,-0.01) node [below] {0.1};
\draw[] (0.2,0.01) -- (0.2,-0.01) node [below] {0.2};
\draw[] (0.3,0.01) -- (0.3,-0.01) node [below] {0.3};
\draw[] (0.4,0.01) -- (0.4,-0.01) node [below] {0.4};
\draw[] (0.5,0.01) -- (0.5,-0.01) node [below] {0.5};
\draw[] (0.01,0.1) -- (-0.01,0.1) node [left] {0.1};
\draw[] (0.01,0.2) -- (-0.01,0.2) node [left] {0.2};
\draw[] (0.01,0.3) -- (-0.01,0.3) node [left] {0.3};
\draw[thick,dotted] (0,0.325) -- (0,0);
\draw[thick,dotted] (0.4,0.325) -- (0.4,0);
\draw[thick,dotted] (0.5,0.325) -- (0.5,0);
\node at (0.45,0.32)[]{$L_1^*=1$};
\draw[latex-latex] (0,0.32) -- node [fill=white,inner sep=1pt]{$L_1^*=2$} (0.4,0.32);

\end{tikzpicture}
    \label{fig:p1}
}\\ 
\subfigure[$\epsilon_1=\frac{1}{2}$, $\epsilon_2=\frac{9}{10}$, $\mu=\frac{1}{10}$]{
\begin{tikzpicture}[xscale=13,yscale=26,font=\scriptsize]
\draw[] (0,0.1) --  (0.4082,0.0918) -- (0.5,0) -- (0,0) -- cycle;
\filldraw [draw=red,fill=red!4!white] (0,0) -- (0,0.1) -- (0.246154,0.088) -- (0.32,0.082426) -- (0.4,0.0724) -- (0.438,0.062) -- (0.5,0);
\filldraw [dashed,draw=green!70!black!80!white,fill=green!70!black!8!white] (0,0) -- (0,0.1) -- (0.5,0) -- cycle;
\filldraw [dotted,draw=blue!80!white,fill=blue!10!white] (0,0) -- (0,0.1) -- (0.1,0) -- cycle;
\draw[dotted,thick] (0.5,0.1325) -- (0.5,0);
\draw[dotted,thick] (0.4,0.1325) -- (0.4,0);
\draw[dotted,thick] (0.32,0.1325) -- (0.32,0);
\draw[dotted,thick] (0.246154,0.1325) -- (0.246154,0);
\draw[dotted,thick] (0,0.1325) -- (0,0);
\node at (0.45,0.13)[]{$L_1^*=1$};
\node at (0.36,0.13)[]{$L_1^*=2$};
\node at (0.282,0.13)[]{$L_1^*=3$};
\draw[latex-latex] (0,0.13) -- node [fill=white,inner sep=1pt]{$L_1^*=4$} (0.246154,0.13);

\draw[] (0,0) rectangle (0.55,0.12);
\node at (0.275,-0.025)[below] {$r_1$};
\node at (-0.05,0.06)[rotate=90,above] {$r_2$};

\draw[] (0.1,0.005) -- (0.1,-0.005) node [below] {0.1};
\draw[] (0.2,0.005) -- (0.2,-0.005) node [below] {0.2};
\draw[] (0.3,0.005) -- (0.3,-0.005) node [below] {0.3};
\draw[] (0.4,0.005) -- (0.4,-0.005) node [below] {0.4};
\draw[] (0.5,0.005) -- (0.5,-0.005) node [below] {0.5};
\draw[] (0.01,0.05) -- (-0.01,0.05) node [left] {0.05};
\draw[] (0.01,0.1) -- (-0.01,0.1) node [left] {0.10};
\end{tikzpicture}
\label{fig:p2}
}\\ 
\subfigure[$\epsilon_1=\frac{1}{2}$, $\epsilon_2=\frac{3}{4}$, $\mu=\frac{9}{10}$]{
\begin{tikzpicture}[xscale=13,yscale=13,font=\scriptsize]
\filldraw[draw=black,fill=white] (0,0.25) -- (0.45455,0.04545) -- (0.5,0) -- (0,0) -- cycle;
\filldraw[draw=red,fill=red!4!white] (0,0.25) --  (0.47619,0.02381) -- (0.5,0)  -- (0,0) -- cycle;
\filldraw[draw=green!70!black,dashed,fill=green!70!black!8!white] (0,0.25) -- (0.5,0) -- (0,0) -- cycle;
\filldraw[draw=blue,dotted,fill=blue!10!white] (0,0.25) -- (0.25,0) -- (0.5,0) -- (0,0) -- cycle;
\draw[] (0,0) rectangle (0.55,0.3);
\node at (0.275,-0.05)[below] {$r_1$};
\node at (-0.05,0.15)[rotate=90,above] {$r_2$};
\draw[] (0.1,0.01) -- (0.1,-0.01) node [below] {0.1};
\draw[] (0.2,0.01) -- (0.2,-0.01) node [below] {0.2};
\draw[] (0.3,0.01) -- (0.3,-0.01) node [below] {0.3};
\draw[] (0.4,0.01) -- (0.4,-0.01) node [below] {0.4};
\draw[] (0.5,0.01) -- (0.5,-0.01) node [below] {0.5};
\draw[] (0.01,0.1) -- (-0.01,0.1) node [left] {0.1};
\draw[] (0.01,0.2) -- (-0.01,0.2) node [left] {0.2};
\draw[] (0.01,0.3) -- (-0.01,0.3) node [left] {0.3};
\draw[thick,dotted] (0,0.325) -- (0,0);
\draw[thick,dotted] (0.5,0.325) -- (0.5,0);
\draw[latex-latex] (0,0.32) -- node [fill=white,inner sep=1pt]{$L_1^*=1$} (0.5,0.32);
\end{tikzpicture}
\label{fig:p3}
}
\caption{Rate regions achieved by different schemes --- Conventional random codes (blue), time-division between separate random codes (green), hybrid coding (red), and genie-aided (non-blind) index coding (white) --- for three different 2-user BICW problems. The number of repetitions used in the hybrid coding scheme is stated along the $x$-axis. (a) For this setting, $L_{max}=2$; (b) For this setting, $L_{max}=4$ and we have emphasized using dashed lines bounds (\ref{eq:ratereg_A}) and (\ref{eq:ratereg_B}) for all $L$ that comprise the boundary of $\mathcal{R}$; (c) For this setting, $L_{max}=1$ and notice even with very little side information, our hybrid coding scheme strictly outperforms conventional schemes.}
    \label{fig:BICW}
\end{figure}
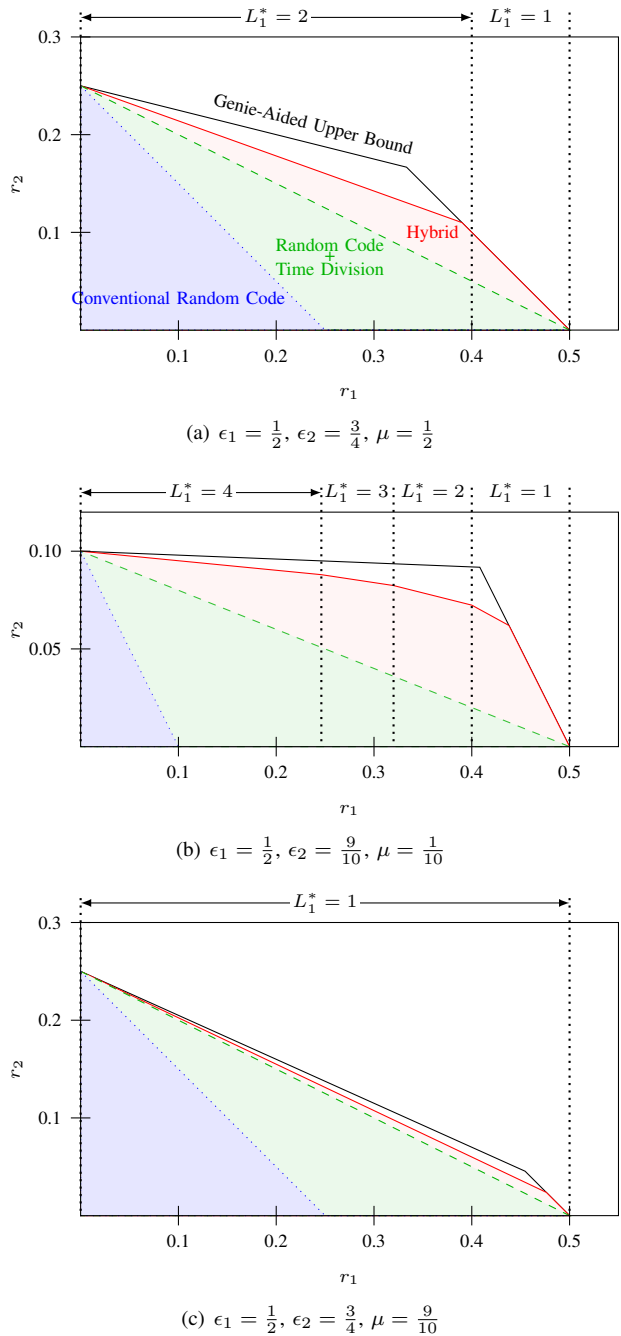

From the scenarios depicted i Figure~\ref{fig:BICW}, we make the following unifying conclusions:
\begin{enumerate}
\item Regardless of the network parameters ($\epsilon_1$, $\epsilon_2$, and $\mu$) hybrid coding always increases the achievable rate region.
\item If we consider a fixed $r_1$, the number of repetitions used in the hybrid encoding scheme increases when User~2 has a weaker channel and more side information (i.e., $\epsilon_2$ grows larger and $\mu$ grows smaller).
\item Hybrid coding can be capacity achieving, as seen on the boundary of the rate regions in all three figures when $r_1$ is close to it's maximum. 
\end{enumerate}

\section{Concluding Remarks}\label{sec:concl}

In this paper, we introduced a generalization of index coding called \emph{blind index coding}, which captures key issues in distributed caching and wireless settings. We demonstrated that the BIC problem introduces novel and interesting challenges that require new analytical tools through three main contributions: 1) we proposed a class of hybrid coding schemes which mix uncoded bits of a subset of messages with randomly linear combinations of other messages, 2) we presented new outer bounds that leveraged a lemma based on a strong data processing to capture the lack of knowledge at the sender, and 3) we demonstrated that in scenarios where the sender-to-user channel is not error-free (specifically, a wireless binary fading channel) repetition of uncoded bits within hybrid codes can further increase the achievable rate.

To further emphasize the importance of analyzing BIC problems, we refer the reader Figure~\ref{fig:2hEBC} which depicts the setting considered~\cite{KMA2014:isit}, which itself was a specific case in the broader class of multiple unicast and multiple multicast problems in wireless erasure networks~\cite{DGPHE2006}. Such problems consider the communication of multiple distinct messages to different users in a wireless network over probabilistic lossy links. 

\begin{figure}[ht]
\centering\vspace{-0.25cm}
    \begin{tikzpicture}[scale=2,font=\footnotesize]
        \node (w1) at (-0.3,0) [anchor=south east] {$\msg_1$};
        \node (w2) at (-0.3,0) [anchor=north east] {$\msg_2$};
        \node (s) at (0,0) [draw,thick] {$\mathsf{S}$};
        \node (r1a) at (1.82,0.7) [draw,circle,inner sep=1pt] {};
        \node (r2a) at (1.82,-0.7) [draw,circle,inner sep=1pt] {};
        \node (r1) at (2,0.7) [draw,thick] {$\mathsf{R}_1$};
        \node (r2) at (2,-0.7) [draw,thick] {$\mathsf{R}_2$};
        \node (d1) at (4.58,0.6) [draw,thick] {$\mathsf{D}_1$};
        \node (d1a) at (4.4,0.7) [draw,circle,inner sep=1pt] {};
        \node (d1b) at (4.4,0.5) [draw,circle,inner sep=1pt] {};
        \node (d2) at (4.58,-0.6) [draw,thick] {$\mathsf{D}_2$};
        \node (d2a) at (4.4,-0.7) [draw,circle,inner sep=1pt] {};
        \node (d2b) at (4.4,-0.5) [draw,circle,inner sep=1pt] {};
        \draw[dotted,thick] (s) -- node[above,pos=0.4]{$\vec{u}$} (0.82,0);
        \draw[-latex,dotted,thick] (0.82,0) --node[sloped,draw,solid,fill=white,inner sep=2pt,rounded corners]{$\epsilon_1$} (r1a) node[anchor=south east]{$\vec{v}_1$};
        \draw[-latex,dotted,thick] (0.82,0) --node[sloped,draw,solid,fill=white,inner sep=2pt,rounded corners]{$\epsilon_1$} (r2a) node[anchor=north east]{$\vec{v}_2$};
        \draw[thick] (r1) -- (2.6,0.7) node[above]{$\vec{x}_1$};
        \draw[-latex,thick] (2.6,0.7) --node[sloped,draw,solid,fill=white,pos=0.4,inner sep=2pt,rounded corners]{$\epsilon_2$} (d1a) node[anchor=south east]{$\vec{y}_{1}$};
        \draw[-latex,thick] (2.6,0.7) --node[sloped,draw,solid,fill=white,pos=0.3,inner sep=2pt,rounded corners]{$\epsilon_3$} (d2b) node[anchor=east]{$\vec{y}_{2}$\ };
        \draw[dotted,thick] (r2) -- (2.6,-0.7);
        \draw[-latex,dotted,thick] (2.6,-0.7) --node[sloped,draw,solid,fill=white,pos=0.3,inner sep=2pt,rounded corners]{$\epsilon_3$} (d1b) node[anchor=north]{\normalsize\textcolor{red}{?}\ \ };
        \draw[-latex,dotted,thick] (2.6,-0.7) --node[sloped,draw,solid,fill=white,pos=0.4,inner sep=2pt,rounded corners]{$\epsilon_2$} (d2a) node[anchor=north]{\normalsize\textcolor{red}{?}\ \ };
        \node (w1h) at (4.8,0.6) [right] {$\widehat{\msg_1}$};
        \node (w2h) at (4.8,-0.6) [right] {$\widehat{\msg_2}$};
        \fill [white!95!red!95!black,opacity=0.475] (1.55,1) -- (1.55,0) -- (3.3,-1) -- (-0.7,-1) -- (-0.7,1) -- cycle;
        \draw [red,very thick,rounded corners] (1.55,1) -- (1.55,0) -- (3.3,-1) -- (5.2,-1) -- (5.2,1) -- cycle;
    \end{tikzpicture}\vspace{-0.25cm}
    \caption{The symmetric two-hop erasure broadcast channel from~\cite{KMA2014:isit} with focus on the embedded BICW problem seen by Relay~1. The network consists of two hops of communication. The first hop is an erasure broadcast channel, whereas the second consists of two parallel, non-interfering erasure broadcasts. Destination~1 wants message $\msg_1$ and Destination~2 wants $\msg_2$, but with no knowledge of erasures, Relay~1 is unaware of the (side) information provided by Relay~2.}\label{fig:2hEBC}
\end{figure}
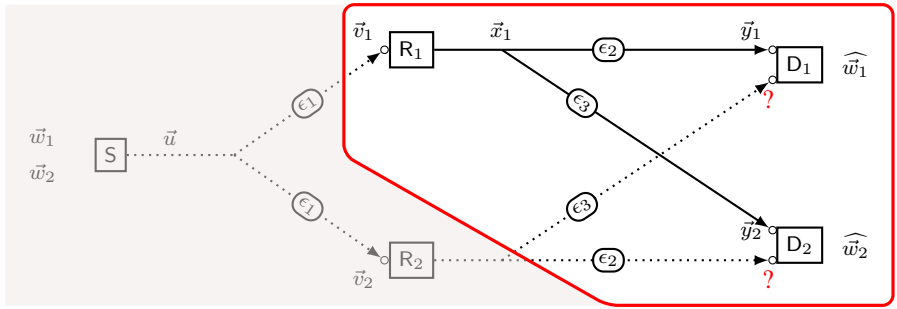

A key contribution of~\cite{KMA2014:isit} was the revelation that it was \emph{strictly suboptimal} for relays within such a network to apply conventional random network coding. Instead, relays imparted structure into their network coded transmissions by XORing \emph{unmixed received bits} of one message with random combinations of another; i.e., relays applied a version of hybrid coding to their received signals in order to outperform conventional random network codes.
From results presented in this work, one arrives at such a relaying strategy naturally. From the point of view of either relay, the transmissions of the other relay are \emph{side information}, and more importantly, due to the lossy nature of links the relay is \emph{blind} as to what side information was provided. Additionally, the transmission model from relays to destinations matches precisely the lossy sender-to-user broadcast considered in Section~\ref{sec:BICW}.

It is important to point out that the BIC and BICW problems in the general setting remains an open problem. Therefore, to conclude the paper we revisit one class of interesting symmetric side information BIC problems (from Section~\ref{sec:num}) that remains unsolved and yet offers a simple and concrete enough case for progress to be made, potentially revealing new insights. 
Consider the following 3-user BIC scenario when side information parameters are pairwise symmetric:
$\mu_{12}=\mu_{21}=a$, $\mu_{13}=\mu_{31}=b$, $\mu_{23}=\mu_{32}=c$ with $a\geq b\geq c$ (for a concrete example we refer the reader to Figure~\ref{fig:q2d}). 
From Theorem~\ref{thm:3uACH}, we find the symmetric achievable rate:
\begin{align}
    r_{sym}={}&\max\left\{\frac{1}{1+a+b+c-ab},\frac{1}{1+a+b}\right\},\label{eq:abc}
\end{align}
and from Theorem~\ref{thm:3uOB}, we have the capacity bound:
\begin{align}
    r_{sym}\leq{}&\frac{1}{1+a+b-\frac{(a-c)(b-c)}{1-c}}.
\end{align}
Notice first that, as in the numerical example of Figure~\ref{fig:q2d}, if $c=0$ or $c=b$ the upper bound is tight and capacity is achieved. However, within the interval $c\in(0,b)$ there exists a gap between achievability and converse. 

Additionally, recall that the first quantity in the max of (\ref{eq:abc}) is the rate achieved by hybrid coding and the second is by conventional random coding. Clearly, hybrid coding provides a rate gain when $c < ab$. This regime is one where the side information Users~2 and 3 have about each others' messages is large and thus Phases~1 and~2 in Figure~\ref{fig:3coding} are small. Our hybrid coding assumes that User~1 ignores these phases, but when they are larger (i.e., as $c$ grows) these transmissions may be used by User~1 to decode messages $\msg_2$ and $\msg_3$. In particular, the case where $c=ab$ (a point notably within the interval $(0,b)$) represents a threshold where the structure of our hybrid code can no longer expect to hide linear subspaces of $\msg_2$ and $\msg_3$ from User~1. 

We conjecture that at this threshold, any method of encoding $\msg_2$ and $\msg_3$ that satisfies the decodability condition at Users~2 and 3 also allows User~1 to decode $\msg_2$ and $\msg_3$ (i.e., at this threshold it is the converse and not achievable scheme that may be tightened).

\bibliographystyle{IEEEtran}
\bibliography{references}

\appendix
\subsection{Evaluating $s^*((\sinfo^\prime,\vec{g}^\prime);(\sinfo_{ij},\vec{g}_{ij}))$}\label{app:sstar}
We now prove that $s^*((\sinfo^\prime,\vec{g}^\prime);(\sinfo_{ij},\vec{g}_{ij}))=\frac{1-\mu_{ij}}{1-\mu^\prime}$ by showing that it may be bounded both from above and below by the same value. We first address the upper bound:
\begin{align}
s^*((&\sinfo^\prime,\vec{g}^\prime);(\sinfo_{ij},\vec{g}_{ij})) \nonumber\\
    \stackrel{(a)}{=}{}& \max_{\ell\in\{1,\ldots,m_j\}} s^*(({\phi}^\prime[\ell],{g}^\prime[\ell]);({\phi}_{ij}[\ell],{g}_{ij}[\ell]))\nonumber\\
    \stackrel{(b)}{\leq}{}& s^*(\overbrace{({\phi}^\prime[1],{g}^\prime[1])}^{\alpha};\overbrace{({\phi}_{ij}[1],{g}_{ij}[1])}^{\beta}).
\end{align}
In step~(a) we apply the tensorization property of $s^*(\cdot)$~\cite{AGKN2014:isit}, and in~(b) we observed that all variables are i.i.d. across $\ell$. To simplify exposition, we now use the following notation: Let $P_\alpha(\cdot)$ and $P_\beta(\cdot)$ denote probability mass functions for $({\phi}^\prime[1],{g}^\prime[1])$ and $({\phi}_{ij}[1],{g}_{ij}[1])$ respectively, and let $Q_\alpha(\cdot)$ and $Q_\beta(\cdot)$ be arbitrary probability mass functions for  $({\phi}^\prime[1],{g}^\prime[1])$ and $({\phi}_{ij}[1],{g}_{ij}[1])$ respectively. Note that the support of both $({\phi}^\prime[1],{g}^\prime[1])$ and $({\phi}_{ij}[1],{g}_{ij}[1])$ is $\{(0,0),(0,1),(1,1)\}$. Using this notation, we now observe
\begin{align}
&s^*(({\phi}^\prime[\ell],{g}^\prime[\ell]);({\phi}_{ij}[\ell],{g}_{ij}[\ell]))\nonumber\\
&={} \sup_{Q_\alpha\neq P_\alpha}\frac{D(Q_\beta || P_\beta)}{D(Q_\alpha || P_\alpha)}\nonumber\\
&={} \sup_{Q_\alpha\neq P_\alpha}\Biggr[
    P_\beta(0,0)\log\left(\frac{P_\beta(0,0)}{Q_\beta(0,0)}\right) 
    + P_\beta(0,1)\log\left(\frac{P_\beta(0,1)}{Q_\beta(0,1)}\right)
    \nonumber\\ 
&\quad\quad + P_\beta(1,1)\log\left(\frac{P_\beta(1,1)}{Q_\beta(1,1)}\right)\Biggr]/
    D(Q_\alpha || P_\alpha)\nonumber\\
&={} \sup_{Q_\alpha\neq P_\alpha}\Biggr[(\delta + (1-\delta)P_\alpha(0,0))\log\left(\frac{\delta + (1-\delta)P_\alpha(0,0)}{\delta + (1-\delta)Q_\alpha(0,0)}\right)
    + (1-\delta)P_\alpha(0,1)\log\left(\frac{(1-\delta)P_\alpha(0,1)}{(1-\delta)Q_\alpha(0,1)}\right)
    \nonumber\\ 
    &\quad\quad + (1-\delta)P_\alpha(1,1)\log\left(\frac{(1-\delta)P_\alpha(1,1)}{(1-\delta)Q_\alpha(1,1)}\right)\Biggr]/
    D(Q_\alpha || P_\alpha)\label{eq:sstareval1}\\
&\stackrel{(c)}{\leq}{} \sup_{Q_\alpha\neq P_\alpha}\Biggr[(1-\delta)P_\alpha(0,0)\log\left(\frac{(1-\delta)P_\alpha(0,0)}{(1-\delta)Q_\alpha(0,0)}\right)
    + (1-\delta)P_\alpha(0,1)\log\left(\frac{(1-\delta)P_\alpha(0,1)}{(1-\delta)Q_\alpha(0,1)}\right)
\nonumber\\ 
&\quad\quad + (1-\delta)P_\alpha(1,1)\log\left(\frac{(1-\delta)P_\alpha(1,1)}{(1-\delta)Q_\alpha(1,1)}\right)\Biggr]/
    D(Q_\alpha || P_\alpha)\nonumber\\
&={} \sup_{Q_\alpha\neq P_\alpha}(1-\delta)\frac{D(Q_\alpha || P_\alpha)}{D(Q_\alpha || P_\alpha)}
={} 1-\delta ={} \frac{1-\mu_{ij}}{1-\mu^\prime}.\label{eq:sstarUB}
\end{align}
Step (c) is verified by observing the following. Trivially we have 
\begin{align}
(\delta + (1-\delta)P_\alpha(0,0))\log\left(\frac{\delta + (1-\delta)P_\alpha(0,0)}{\delta + (1-\delta)Q_\alpha(0,0)}\right) 
&\leq
\max_{y\in[0,1]} (y + (1-\delta)P_\alpha(0,0))\log\left(\frac{y + (1-\delta)P_\alpha(0,0)}{y + (1-\delta)Q_\alpha(0,0)}\right).\label{eq:stepcsstar}
\end{align}
Since for any positive $y$, $U$, and $V$:
\begin{align*}
\frac{\partial}{\partial y}(y+U)\log\left(\frac{y+U}{y+V}\right) 
={}& 1+\log\left(\frac{y+U}{y+U}\right) - \frac{y+U}{y+V}\\
\leq{}& 0,
\end{align*}
we observe by letting $U=(1-\delta)P_\alpha(0,0)$ and $V=(1-\delta)Q_\alpha(0,0)$ that the right hand side of (\ref{eq:stepcsstar}) attains its maximum at $y=0$.

To show that $s^*((\sinfo^\prime,\vec{g}^\prime);(\sinfo_{ij},\vec{g}_{ij}))\geq\frac{1-\mu_{ij}}{1-\mu^\prime}$, we restrict the domain of $Q_\alpha$ to mass functions where $Q_{\alpha}(0,0)=P_\alpha(0,0)$ and observe from (\ref{eq:sstareval1})
\begin{align}
&s^*(({\phi}^\prime[\ell],{g}^\prime[\ell]);({\phi}_{ij}[\ell],{g}_{ij}[\ell]))\nonumber\\
&={} \sup_{Q_\alpha\neq P_\alpha}\Biggr[(\delta + (1-\delta)P_\alpha(0,0))\log\left(\frac{\delta + (1-\delta)P_\alpha(0,0)}{\delta + (1-\delta)Q_\alpha(0,0)}\right)
     + (1-\delta)P_\alpha(0,1)\log\left(\frac{(1-\delta)P_\alpha(0,1)}{(1-\delta)Q_\alpha(0,1)}\right)
\nonumber\\ 
&\quad\quad + (1-\delta)P_\alpha(1,1)\log\left(\frac{(1-\delta)P_\alpha(1,1)}{(1-\delta)Q_\alpha(1,1)}\right)\Biggr]/
    D(Q_\alpha || P_\alpha)\nonumber\\
&\geq{} \sup_{\substack{Q_\alpha\neq P_\alpha\\Q_\alpha(0,0)= P_\alpha(0,0)}}\Biggr[ 
(1-\delta)P_\alpha(0,1)\log\left(\frac{(1-\delta)P_\alpha(0,1)}{(1-\delta)Q_\alpha(0,1)}\right)
    + (1-\delta)P_\alpha(1,1)\log\left(\frac{(1-\delta)P_\alpha(1,1)}{(1-\delta)Q_\alpha(1,1)}\right)\Biggr]/
    D(Q_\alpha || P_\alpha)\nonumber\\
&={} \sup_{\substack{Q_\alpha\neq P_\alpha\\Q_\alpha(0,0)= P_\alpha(0,0)}}(1-\delta)\frac{D(Q_\alpha || P_\alpha)}{D(Q_\alpha || P_\alpha)}
={} 1-\delta ={} \frac{1-\mu_{ij}}{1-\mu^\prime}.\label{eq:sstarLB}
\end{align}
\begin{remark}
Note that the validity of Lemma~\ref{lem:split} only requires the upper bound (\ref{eq:sstarUB}). However, by evaluating the lower bound (\ref{eq:sstarLB}) as well, we may confirm the exact value of $s^*((\sinfo^\prime,\vec{g}^\prime);(\sinfo_{ij},\vec{g}_{ij}))$. This value has an intuitive interpretation as the success probability of the channel that takes each bit of the virtual signal $\sinfo^\prime$ as input and gives $\sinfo_{ij}$ as output.
\end{remark}
\subsection{Proof of (\ref{eq:3u1})}\label{app:3u2}
If $\mu_{kj} \geq \mu_{ij}$, we observe that $\msg_{i}$ and $(\sinfo_{ij},\vec{g}_{ij})$ are statistically enhanced versions of  $(\sinfo_{ki},\vec{g}_{ki})$ and $(\sinfo_{kj},\vec{g}_{kj})$ respectively. We may further enhance the side information of User~$k$ by also providing $(\sinfo_{ik},\vec{g}_{ik})$. Applying Fano's inequality at User~$k$ with side information enhancement, we find
\begin{align}
    nr_k 
    \leq{}& I(\vec{x},\sinfo_{ki},\vec{g}_{ki},\sinfo_{kj},\vec{g}_{kj};\msg_k) +o(n)\nonumber\\
    \leq{}& I(\vec{x},\msg_{i},\sinfo_{ij},\vec{g}_{ij},\sinfo_{ik},\vec{g}_{ik};\msg_k) +o(n)\nonumber\\
    ={}& I(\sinfo_{ik},\vec{g}_{ik};\msg_k) +H(\vec{x}|\msg_{i},\sinfo_{ij},\vec{g}_{ij},\sinfo_{ik},\vec{g}_{ik}) 
         - H(\vec{x}|\msg_{i},\sinfo_{ij},\vec{g}_{ij},\msg_k) + o(n)\nonumber\\ 
    ={}& n(1-\mu_{ik})r_k +H(\vec{x}|\msg_{i},\sinfo_{ij},\vec{g}_{ij},\sinfo_{ik},\vec{g}_{ik})
        - H(\vec{x}|\msg_{i},\sinfo_{ij},\vec{g}_{ij},\msg_k) + o(n)\nonumber\\ 
    \stackrel{(a)}{\leq}{}& n(1-\mu_{ik})r_k +H(\vec{x}|\msg_{i},\sinfo_{ij},\vec{g}_{ij},\sinfo_{ik},\vec{g}_{ik}) 
        - \mu_{ij}H(\vec{x}|\msg_{i},\msg_k) + o(n),\label{eq:3uproof1a}
\end{align}
where in step (a) we applied (\ref{eq:lem:split:1}) from Lemma~\ref{lem:split}, by letting $V=(\msg_i,\msg_k)$.

By combining (\ref{eq:3uprooffano}) and scaled versions of (\ref{eq:3uproof1a}) and (\ref{eq:3uproof2B2}), and taking the limit as $n$ grows large, we arrive at (\ref{eq:3u1}).  Similarly, if $\mu_{jk} \geq \mu_{ik}$ we may switch the roles of Users~$j$ and $k$ in the above analyses to arrive at a similar conclusion.

\subsection{Proof of Theorem~\ref{thm:KuOB}}\label{app:KuOB}

We now formally prove Theorem~\ref{thm:KuOB}:
The following notation and claim will simplify exposition of the proof. Let $\vec{g}_j^{(\eta)}$ be a length-$m_j$ vector of i.i.d. Bernoulli random variables that take a value of 0 with probability $\eta$, and let
\begin{align*}
    \phi_j^{(\eta)}[\ell] = g_j^{(\eta)}[\ell]w_j[\ell].
\end{align*}
Notice, for instance, that $\sinfo_j^{(\mu_{i,j})}$ is statistically equivalent to $(\sinfo_{i,j},\vec{g}_{i,j})$, and that $\vec{\psi}_j^{0}$ and $\vec{\psi}_j^{1}$ are equal to $(\msg_j,\vec{1})$ and $(\vec{0},\vec{0})$ respectively. Additionally, we define $\vec{\psi}_j^{(\eta)} \triangleq (\sinfo_j^{(\eta)},\vec{g}^{(\eta)})$.
We now formalize the notion of statistically enhanced side information, with the following claim, which is consequence of the sender being blind to the precise side information:
\begin{claim}\label{cl:staten}
Let $\eta_1>\eta_2$ be given. For any $k\in\{1,\ldots,K\}$ and $V$ independent of $\msg_j$, $\vec{g}_{j}^{(\eta_1)}$, and $\vec{g}_{j}^{(\eta_1)}$, we have
\begin{align}
I(\vec{x},\vec{\psi}_j^{(\eta_1)},V;\msg_k) \leq I(\vec{x},\vec{\psi}_j^{(\eta_2)},V;\msg_k),\label{eq:cl:statenh1}
\end{align}
and
\begin{align}
H(\vec{x}|\vec{\psi}_j^{(\eta_1)},V) \geq H(\vec{x}|\vec{\psi}_j^{(\eta_2)},V).\label{eq:cl:statenh2}
\end{align}
\end{claim}
\begin{proof}
The proof is an immediate consequence of the sender being blind to the side information channels: since each quantity is a function only of the marginal distribution of $\vec{g}_j^{(\eta)}$ we may in fact define $\vec{\psi}_j^{(\eta_1)}$ as a physically degraded version of $\vec{\psi}_j^{(\eta_1)}$. Hence, the right side of (\ref{eq:cl:statenh1}) can be seen as the mutual information between a message and an enhanced channel output, and the right side of (\ref{eq:cl:statenh2}) can be seen as an enhance signal adding conditioning.
\end{proof}

The proof now proceeds as follows. At each node $(\ell,i)$ in the OBT with $\ell<K$, we will apply Fano's inequality to a virtual user that desires message $\msg_{v[\ell,i]}$. This virtual user is given side information signals that are statistically enhanced versions of $\{\sinfo_{v[\ell,i],j}\}$ (i.e. the actual user in the BIC problem). The statistical properties of the side information channels governed by $\{\eta_j[\ell,i]\}_j$ as defined by the OBT structure and (\ref{eq:KuOBrec4a}). We will see that by applying Lemma~\ref{lem:split} to the expansion of Fano's inequality at each node, and scaling the resulting scaling expressions according, that terms on the right hand side of each expression will cancel and we will arrive at the stated bound. 

We denote the complete collection of side information and channel state information given to a virtual user represented by the $i$-th node in level $\ell$ of the OBT as $\vec{\Psi}[\ell,i] = (\vec{\psi}_1^{\eta_1[\ell,i]},\ldots,\vec{\psi}_K^{\eta_K[\ell,i]})$. Similarly, we denote the collection of side information/channel state given to the actual User~$v[\ell,i]$ as $\vec{\Psi}_{v[\ell,i]}$.

If $i$ is odd, recall from (\ref{eq:KuOBrec4a}) that $\eta_{v[\ell,i]}[\ell,i] = 1$ which implies that none of the virtual user's desired message is provided as (enhanced) side information. Thus, for off $i$ and $\ell<K-1$ we have
\begin{align}
   n r_{v[\ell,i]} \leq{}& I(\vec{x},\vec{\Psi}_{v[\ell,i]};\msg_{v[\ell,i]}) +o(n)\nonumber\\
    \leq{}& I(\vec{x},\vec{\Psi}[\ell,i];\msg_{v[\ell,i]}) +o(n)\nonumber\\
    ={}& H(\vec{x}|\vec{\Psi}[\ell,i]) - H(\vec{x}|\vec{\Psi}[\ell,i],\msg_{v[\ell,i]}) +o(n)\nonumber\\ 
    \stackrel{(a)}{\leq}{}& H(\vec{x}|\vec{\Psi}[\ell,i]) - \zeta[\ell,i]H(\vec{x}|\vec{\Psi}[\ell+1,2i-1])
        - (1-\zeta[\ell,i])H(\vec{x}|\vec{\Psi}[\ell+1,2i])+o(n),\label{eq:KuOBproof1}
\end{align}
where in step (a) we applied a combination of Lemma~\ref{lem:split} and observing from (\ref{eq:KuOBrec4a}) that $\vec{\Psi}[\ell+1,2i-1]$ or $\vec{\Psi}[\ell+1,2i]$ can only increase conditioning relative to $\vec{\Psi}[\ell,i],\msg_{v[\ell,i]}$.

If $i$ is even, some of the virtual user's desired message may have been provided as side information. Thus, for even $i$ and $\ell<K-1$ we have
\begin{align}
   n r_{v[\ell,i]} \leq{}& I(\vec{x},\vec{\Psi}_{v[\ell,i]};\msg_{v[\ell,i]}) +o(n)\nonumber\\
    \leq{}& I(\vec{x},\vec{\Psi}[\ell,i];\msg_{v[\ell,i]}) +o(n)\nonumber\\
    ={}& I(\vec{\psi}_{v[\ell,i]}^{\eta_{v[\ell,i]}[\ell,i]};\msg_{v[\ell,i]})
        + H(\vec{x}|\vec{\Psi}[\ell,i]) - H(\vec{x}|\vec{\Psi}[\ell,i],\msg_{v[\ell,i]})+o(n)\nonumber\\ 
    ={}& n (1-\eta_{v[\ell,i]}[\ell,i])r_{v[\ell,i]}
        + H(\vec{x}|\vec{\Psi}[\ell,i]) - H(\vec{x}|\vec{\Psi}[\ell,i],\msg_{v[\ell,i]})+o(n)\nonumber\\ 
    \stackrel{(b)}{=}{}& n \left(1-\eta_{v[\ell,i]}\left[\ell-1,\left\lceil\frac{i}{2}\right\rceil\right]\right)r_{v[\ell,i]}
        + H(\vec{x}|\vec{\Psi}[\ell,i]) - H(\vec{x}|\vec{\Psi}[\ell,i],\msg_{v[\ell,i]})+o(n)\nonumber\\ 
    \stackrel{(c)}{\leq}{}& n \left(1-\eta_{v[\ell,i]}\left[\ell-1,\left\lceil\frac{i}{2}\right\rceil\right]\right)r_{v[\ell,i]}
        +H(\vec{x}|\vec{\Psi}[\ell,i]) \nonumber\\
        & - \zeta[\ell,i]H(\vec{x}|\vec{\Psi}[\ell+1,2i-1])
        - (1-\zeta[\ell,i])H(\vec{x}|\vec{\Psi}[\ell+1,2i]),+o(n).\label{eq:KuOBproof2}
\end{align}
In (b) we applied (\ref{eq:KuOBrec4a}), and in (c) (a) we applied a combination of Lemma~\ref{lem:split} and observed from (\ref{eq:KuOBrec4a}) that $\vec{\Psi}[\ell+1,2i-1]$ or $\vec{\Psi}[\ell+1,2i]$ can only increase conditioning relative to $\vec{\Psi}[\ell,i],\msg_{v[\ell,i]}$.

We now address the base cases (i.e., when $\ell=K-1$). We first observe that at the $K-1$ level, because all paths from root to leaf are permutation of all user indices, if $j\notin \{v[K-1,i],v[K,i]\}$ then $\eta_j[K-1,i]=0$. Equivalently, if $j\notin \{v[K-1,i],v[K,i]\}$ then $\vec{\psi}_{j}^{\eta_j[K-1,i]} = (\msg_j,\vec{1})$. 

Now recall that if $i$ is odd, then $\eta_{v[K-1,i]}[K-1,i] = 1$, and from Fano we have
\begin{align}
   n r_{v[K-1,i]} \leq{}& I(\vec{x},\vec{\Psi}_{v[K-1,i]};\msg_{v[K-1,i]}) +o(n)\nonumber\\
    \leq{}& I(\vec{x},\vec{\Psi}[K-1,i];\msg_{v[K-1,i]}) +o(n)\nonumber\\
    ={}& H(\vec{x}|\vec{\Psi}[K-1,i])
        - H(\vec{x}|\vec{\psi}_{v[K,i]}^{(\eta_{v[K,i]}[K-1,i])},\{\msg_j\}_{j\neq v[K,i]}) +o(n)\nonumber\\
    \stackrel{(d)}{\leq}{}& H(\vec{x}|\vec{\Psi}[K-1,i])
        - \eta_{v[K,i]}[K-1,i]H(\vec{x}|\{\msg_j\}_{j\neq v[K,i]}) +o(n),\label{eq:KuOBproof3}
\end{align}
where in step (d) we applied (\ref{eq:lem:split:1}). If $i$ is even, then $\eta_{v[K-1,i]}[K-1,i]$ can be less than 1, and from Fano we have
\begin{align}
   n r_{v[K-1,i]} \leq{}& I(\vec{x},\vec{\Psi}_{v[K-1,i]};\msg_{v[K-1,i]}) +o(n)\nonumber\\
    \leq{}& I(\vec{x},\vec{\Psi}[K-1,i];\msg_{v[K-1,i]}) +o(n)\nonumber\\
    ={}& I(\vec{\psi}_{v[K-1,i]}^{\eta_{v[K-1,i]}[K-1,i]};\msg_{v[K-1,i]})
        + H(\vec{x}|\vec{\Psi}[K-1,i])
        - H(\vec{x}|\vec{\psi}_{v[K,i]}^{(\eta_{v[K,i]}[K-1,i])},\{\msg_j\}_{j\neq v[K,i]}) +o(n)\nonumber\\
    ={}& n(1-\eta_{v[K-1,i]}[K-1,i])r_{v[K-1,i]}
        + H(\vec{x}|\vec{\Psi}[K-1,i]) 
        - H(\vec{x}|\vec{\psi}_{v[K,i]}^{(\eta_{v[K,i]}[K-1,i])},\{\msg_j\}_{j\neq v[K,i]}) +o(n)\nonumber\\
    ={}& n\left(1-\eta_{v[K-1,i]}\left[K-2,\left\lceil\frac{i}{2}\right\rceil\right]\right)r_{v[\ell,i]}
        + H(\vec{x}|\vec{\Psi}[K-1,i])
        - H(\vec{x}|\vec{\psi}_{v[K,i]}^{(\eta_{v[K,i]}[K-1,i])},\{\msg_j\}_{j\neq v[K,i]}) +o(n)\nonumber\\
    \stackrel{(e)}{\leq}{}& n\left(1-\eta_{v[K-1,i]}\left[K-2,\left\lceil\frac{i}{2}\right\rceil\right]\right)r_{v[K-1,i]}
        + H(\vec{x}|\vec{\Psi}[K-1,i])
        - \eta_{v[K,i]}[K-1,i]H(\vec{x}|\{\msg_j\}_{j\neq v[K,i]}) +o(n),\label{eq:KuOBproof4}
\end{align}
where in (e) we applied (\ref{eq:lem:split:1}).

Finally, at the $K$-th level of the OBT we have trivially
\begin{align}
   n r_{v[K,i]} \leq{}& I(\vec{x},\vec{\Psi}_{v[K,i]};\msg_{v[K,i]}) +o(n)\nonumber\\
       \leq{}& I(\vec{x},\{\msg_j\}_{j\neq v[K,i]};\msg_{v[K,i]}) +o(n)\nonumber\\
       \leq{}& H(\vec{x}|\{\msg_j\}_{j\neq v[K,i]}) +o(n).\label{eq:KuOBproof5}
\end{align}

Using (\ref{eq:KuOBproof1})--(\ref{eq:KuOBproof5}), and scaling expressions according the the coefficients $\zeta[\ell,i]$ at each node, we recursively arrive at
\begin{align*}
n\Gamma_A[1,1] 
    \leq{}& H(\vec{x}|\vec{\Psi}_{v[\ell,i]})+o(n)\\
    \leq{}& n+o(n),
\end{align*}
as desired.
\hfill\qed


\subsection{Proof of Lemma~\ref{lem:achieveBICW}}
\label{app:lem:achieveBICW}
Consider $(\rho_1,L_1,\alpha_1)$ fixed and rate pair $(r_1,r_2)$ such that (\ref{eq:achJLA0})--(\ref{eq:achJLA3}) are satisfied. To prove the lemma, we demonstrate that there exists a scheme and that using this scheme $(r_1,r_2)$ is achievable according to Definition~\ref{sec:problem}. To do so, we first show that if (\ref{eq:achJLA0}) is satisfied, then a sequence of $\mathbf{U}_1$ matrices (and by proxy encoding functions) exists. We then argue that using the described encoding and decoding strategies with $(\rho_2,L_2,\alpha_2)=(1,0,0)$, the usual equation counting argument applied to conventional random codes in erasure channels suffices to prove that error probability vanishes as $n\rightarrow\infty$.

If (\ref{eq:achJLA0}) is satisfied, then for every $n$, we may choose $m_1^{(n)}=\lfloor nr_1\rfloor$, thereby satisfying (\ref{eq:feasible}) and guaranteeing the existence of a RRP matrix for $\msg_1$. Furthermore, we see that $\lim_{n\rightarrow\infty}\frac{m_1^{(n)}}{n} = r_1$. Bearing this in mind, we let $m_2^{(n)}=\lfloor n r_2\rfloor$ for each $n$ and consider the scheme that uses RRP matrice with parameters $(\rho_1,L_1,\alpha_1)$ and $(\rho_2,L_2,\alpha_2)$  to encode $\msg_1$ and $\msg_2$ respectively.

Recall that in our decoding strategy, $\msg_2$ is decoded at each user by first extracting clean equations of $\msg_2$ (where by clean, we mean the contribution of $\msg_1$ can be canceled out as explained previously). Then, the next step is decoding $\msg_2$ from the clean equations. We now calculate the probability that at time $t$ User~$i$, where $i=1,2$, can extract a clean equation of $\msg_2$ for three cases:
\begin{itemize}
\item If $t = (\ell-1)m_1^{(n)} + k$, where $\ell$ is an positive integer with $\ell\leq L_1$ and $1\leq k \leq m_1^{(n)}$, then $\mathbf{U}_1(t,:)\msg_1$ is the $\ell$-th repetition of $\msg_1[k]$. The probability of receiving a clean equation at User~1 is
$\eta_1(t) = (1-\epsilon_1)(1-\epsilon_1^{\ell-1})$, and at User~2 is $\eta_2(t) = (1-\epsilon_2)(1-\mu\epsilon_2^{\ell-1})$,
which are the products of probabilities that User~$i$ at time $t$ receives an unerased transmission and has previously received a transmission (or side information in at User~2) containing the same $\msg_1[k]$.
\item If $L_1m_1^{(n)}<t\leq L_1m_1^{(n)}+\rho_1n$, then $\mathbf{U}_1(t,:)\msg_1$ is random combination of $\msg_1$ bits, and $\eta_1(t) = \eta_2(t) = 0$ (i.e., to decode $\msg_2$, each user ignores these transmissions).
\item If $L_1m_1^{(n)}+\rho_1n<t$, then $\mathbf{U}_1(t,:)\msg_1=0$, and the probability of a clean equation is the probability of an unerased transmission: $\eta_1(t) = 1-\epsilon_1$ and $\eta_2(t) = 1-\epsilon_2$.
\end{itemize}

From capacity analysis of point-to-point erasure channels, we note that a random linear coded message, $m_2$, with rate $r_2$ is decodable at User~1 with arbitrarily low probability of error as $n\rightarrow\infty$ if the number of received random linear equations of $\msg_2$ is sufficiently large. Specifically, by counting the number of clean random linear equations of $\msg_2$ received by User~1, we see that $r_2$ must satisfy:
\begin{align}
    r_2 &{}\leq \lim_{n\rightarrow\infty}\frac{1}{n}\sum_{t=1}^{n} \eta_1[t]\nonumber\\
        ={}& \lim_{n\rightarrow\infty} \frac{1}{n}\left(\sum_{k=1}^{\alpha_1 m_1^{(n)}}\sum_{\ell=1}^{L_1+1} (1-\epsilon_1)(1 - \epsilon_1^{\ell-1})
            +\sum_{k=\alpha_1 m_1^{(n)}+1}^{m_1^{(n)}}\sum_{\ell=1}^{L_1} (1-\epsilon_1)(1 - \epsilon_1^{\ell-1})
            + (n-L_1m_1^{(n)}-\alpha_1 m_1^{(n)}-\rho_1n)(1-\epsilon_1)\right)\nonumber\\
        ={}& (1-\epsilon_1)(1 - \rho_1) - r_1\left[1-\epsilon_1^{L_1}+\alpha_1 (\epsilon_1^{L_1}-\epsilon_1^{L_1+1}) \right].
        \label{eq:r2int1}
\end{align}
Through analogous analysis, we find that communicating of $\msg_2$ to User~2 is possible with arbitrarily low error probability if
\begin{align}
    r_2 \leq{}& (1-\epsilon_2)(1-\rho_1) - r_1\mu\left[1-\epsilon_2^{L_1}+\alpha_1 (\epsilon_2^{L_1}-\epsilon_2^{L_1+1}) \right].
        \label{eq:r2int2}
\end{align}
Notice that (\ref{eq:r2int1}), (\ref{eq:r2int2}) are equivalent to (\ref{eq:achJLA2}), (\ref{eq:achJLA3}) and thus if both expressions are satisfied, then $\msg_2$ is decodable at each user with high probability.

We now address achievability of $r_1$ assuming that User~1 has already successfully decoded and canceled $\msg_2$ from its received signal. It is sufficient to show that $r_1$ satisfying any one of (\ref{eq:achJLA0})--(\ref{eq:achJLA3}) is achievable, and we do so by proving achievability of $r_1$ satisfying (\ref{eq:achJLA1}).  
Observe that the repetition portion of the RRP matrix supplies through the erasure channel a subset of message bits to User~1. Even if the repetition-code-supplied bits are removed, note that the random linear code portion of the RRP matrix still represents a random linear code applied to bits unknown after repetition. To decode, the total number of bits received though repetition and linearly independent equations received must be equal to $m_1^{(n)}$. 

Using Hoeffding's inequality~\cite{Hoeffding1963}, one can show that with high probability as $n\rightarrow\infty$ the number of bits received through repetition coding is concentrated around its mean, 
\begin{align*}
\sum_{k=1}^{\alpha_1 m_1^{(n)}} \epsilon_1^{L_1+1} + \sum_{k=\alpha_1 m_1^{(n)}+1}^{m_1^{(n)}} \epsilon_1^{L_1} 
=& m_1^{(n)}\left(\epsilon_1^{L_1} - \alpha_1(\epsilon_1^{L_1}-\epsilon_1^{L_1+1})\right). 
\end{align*}
Therefore, after the repetition phase, approximately $m_1^{(n)}\left(1-\epsilon_1^{L_1} + \alpha_1(\epsilon_1^{L_1}-\epsilon_1^{L_1+1})\right)$ bits remain to be communicated using the random linear coding phase.

Through the usual argument that random linear combinations are independent w.h.p. as $n\rightarrow\infty$, the random coding portion of the scheme supplies approximately $\rho_1n(1-\epsilon_1)$ equations.
%
%
Therefore, as $n\rightarrow\infty$, we may expect the random linear coding portion to resolve all message bits of $\msg_1$ that were not received during the repetition phase if
$m_1^{(n)}\left(1-\epsilon_1^{L_1} + \alpha_1(\epsilon_1^{L_1}-\epsilon_1^{L_1+1})\right)\leq \rho_1n(1-\epsilon_1)$, or simply a rate $r_1$ is achievable if it satisfies
$r_1={} \lim_{n\rightarrow\infty}\frac{m_1^{(n)}}{n} \leq \rho_1\frac{1-\epsilon_1}{\epsilon_1^{L_1}-\alpha_1(\epsilon_1^{L_1}-\epsilon_1^{L_1+1})}$.\hfill\qed

\end{document}